\documentclass[authoryear,preprint,review,12pt]{elsarticle}

\usepackage{amssymb}

\usepackage[T1]{fontenc}

\usepackage{natbib}

\usepackage[dvips]{geometry}
\usepackage{amssymb,amsmath,latexsym}
\usepackage{amsthm}
\usepackage{pst-node}
\usepackage{xcolor}
\usepackage{wrapfig}
\newpsstyle{Cblue}{fillstyle=solid,fillcolor=blue!30}
\newpsstyle{Cred}{fillstyle=solid,fillcolor=red!30,shadow=true}

\usepackage{pstricks}    
\usepackage{graphicx}    
\usepackage{multido}

\usepackage{algpseudocode}
\usepackage{algorithm}
\usepackage{float}
\usepackage{enumitem}
\setlist[enumerate,1]{label=\arabic*., ref=\arabic*}
\setlist[enumerate,2]{label=\emph{\alph*}), ref=\theenumi.\emph{\alph*}}
\setlist[enumerate,3]{label=\roman*), ref=\theenumii.\roman*}

\newtheorem{theorem}{Theorem}
\newtheorem{definition}{Definition}
\newtheorem{lemma}{Lemma}

\newcommand{\baseh}{
     \SpecialCoor
     \psset{shadow=true}
     \psset{radius=4mm}
     \psset{linecolor=black}
     \psset{linewidth=0.3mm}
     \psset{linestyle=solid}
     \Cnode(0;0){a}
     \Cnode(2.5;0){b}
     \Cnode(2.5;60){c}
     \Cnode(2.5;120){d}
     \Cnode(2.5;180){e}
     \Cnode(2.5;240){f}
     \Cnode(2.5;300){g}
     \psset{nodesep=2mm}
     \psset{shadow=false}
     \psset{linestyle=dashed}
     \ncline{-}{a}{b}
     \ncline{-}{a}{c}
     \ncline{-}{a}{d}
     \ncline{-}{a}{e}
     \ncline{-}{a}{f}
     \ncline{-}{a}{g}
     \ncline{-}{b}{c}
     \ncline{-}{c}{d}
     \ncline{-}{d}{e}
     \ncline{-}{e}{f}
     \ncline{-}{f}{g}
     \ncline{-}{g}{b}
     \psset{linestyle=solid}
}
\newcommand{\baseE}{
    \baseh
    \psset{linecolor=gray}
    \psset{linewidth=3mm}
    \psset{linestyle=solid}
    \psset{nodesep=-2mm}
    \ncline{-}{b}{c}
    \ncline{-}{c}{d}
    \ncline{-}{d}{e}
    \ncline{-}{f}{g}
    \ncline{-}{a}{d}
}

\def\polypmIIId#1{\pspolygon[linestyle=none,fillstyle=solid,fillcolor=black]}

\journal{arXiv.org}

\begin{document}

\begin{frontmatter}



\title{Linking and Cutting Spanning Trees}


\author[label1]{Lu\'{i}s~M.~S.~Russo\corref{lsr}}
\author[label1]{Andreia Sofia Teixeira}
\author[label1]{Alexandre~P.~Francisco}
\address[label1]{INESC-ID and the Department of Computer Science and
  Engineering, \\ Instituto Superior T\'{e}cnico, Universidade de
  Lisboa.}
\cortext[lsr]{Corresponding author.}

\begin{abstract}
%
  We consider the problem of uniformly generating a spanning tree, of a
  connected undirected graph. This process is useful to compute statistics,
  namely for phylogenetic trees. We describe a Markov chain for producing
  these trees. For cycle graphs we prove that this approach
  significantly outperforms existing algorithms. For general graphs we
  obtain no analytical bounds, but experimental results show that the chain
  still converges quickly. This yields an efficient algorithm, also due to
  the use of proper fast data structures. To bound the mixing time of the
  chain we describe a coupling, which we analyse for cycle graphs and
  simulate for other graphs.
\end{abstract}

\begin{keyword}
%
Spanning Tree \sep Uniform Generation \sep Markov Chain \sep Mixing Time
\sep Link Cut Tree


\PACS 02.10.Ox  \sep 02.50.Ga \sep 02.50.Ng \sep 02.70.Uu


\MSC[2010] 05C81 \sep 05C85 \sep 60J10 \sep 60J22 \sep 65C40  \sep 68R10

\end{keyword}

\end{frontmatter}



\section{Introduction}
\label{sec:intro}
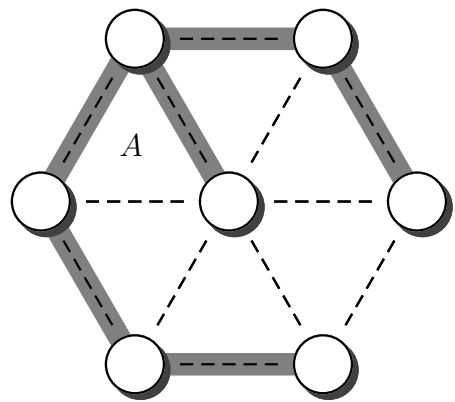
\begin{wrapfigure}{r}{0pt}
  \centering
  \begin{pspicture}(-2.9,-2.8)(3.0,2.8)
    \baseE
    \ncline{-}{e}{f}
    \baseh
    \rput(1.5;150){$A$}
  \end{pspicture}
  \caption{ A Spanning tree $A$ over a graph $G$.}
\label{fig:Intro}
\end{wrapfigure}
A spanning tree $A$ of an undirected connected graph $G$ is a tree, i.e., a
connected set of edges without cycles, that spans every vertex of
$G$. Every vertex of $G$ occurs in some edge of $A$. Figure~\ref{fig:Intro}
shows an example. The vertexes of the graph are represented by circles, the
set of vertexes is denoted by $V$. The edges of $G$ are represented by
dashed lines, the set of edges is represented by $E$. The edges of the
spanning tree $A$ are represented by thick grey lines.
We also use $V$ and $E$ to mean respectively the
size of the set $V$ and the size of set $E$, i.e., the number of vertexes and the number of edges.
In case the expression
can be interpreted as a set, instead of a number, we avoid the ambiguity by
writing $|V|$ and $|E|$, respectively. 

We aim to compute one of such spanning trees $A$, uniformly among all possible
spanning trees. The number of these trees may vary significantly, from $1$
to $V^{V-2}$, depending on the underlying graph~\citep[Chapter
22]{Borchardt,Cayley,aigner2010proofs}. Computing such a tree uniformly and efficiently is
challenging for several reasons: the number of such trees is usually
exponential; the structure of the resulting trees is largely heterogeneous,
as the underlying graphs change. The contributions of this paper are the
following:
\begin{itemize}
\item We present a new algorithm, which given a graph $G$, generates a
  spanning tree of $G$ uniformly at random. The algorithm uses the link-cut
  tree data structure to compute randomizing operations in $O(\log V)$
  amortized time per operation. Hence, the overall algorithm takes
  $O(\tau \log V)$ time to obtain an uniform spanning tree of
  $G$, where $\tau$ is the mixing time of a Markov chain that is
  dependent on $G$. Theorem~\ref{teo1} summarizes this result.

\item We propose a coupling to bound the mixing time $\tau$. The analysis
  of the coupling yields a bound for cycle graphs,
  Theorem~\ref{teo:cycleC}, and for graphs which consists of simple cycles connect by bridges or
  articulation points, Theorem~\ref{teo:cycleCpp}. We also simulate this procedure experimentally
  to obtain bounds for other graphs. The link-cut tree data structure is
  also key in this process. Section~\ref{sec:experimental-results} shows
  experimental results, including other classes of graphs.
  \end{itemize}
  The structure of the paper is as follows. In Section~\ref{sec:problem}
  we introduce the problem and explain its subtle nature. In
  Section~\ref{sec:idea} we explain our approach and point out that using
  the link cut tree data structure is much faster than repeating DFS
  searches. In Section~\ref{sec:details} we thoroughly justify our results,
  proving that the underlying Markov chain has the necessary properties and
  providing experimental results of our algorithm. In
  Section~\ref{sec:related-work} we describe the related work concerning
  random spanning trees, link cut trees and mixing time of Markov
  chains. In Section~\ref{sec:conclusions} we present our conclusions.
\section{The Challenge}
\label{sec:problem}
We start by describing an intuitive process for generating
spanning trees, that does not obtain a uniform distribution. Therefore it
produces some trees with a higher probability than others. This serves to
illustrate that the problem is harder than it may seem at first
glance. Moreover we explain why this process is biased, using a counting
argument.

A simple procedure to build $A$ consists in using a Union-Find data
structure~\citep*{Galler:1964:IEA:364099.364331}, to guarantee that $A$ does
not contain a cycle. Note that these structures are strictly incremental,
meaning that they can be used to detect cycles but can not be used to
remove an edge from the cycle. Therefore the only possible action is to
discard the edge that creates the cycle.

Let us analyse a concrete example of the resulting distribution of spanning
trees. We shall show that this distribution is not uniform. First generate
a permutation $p$ of $E$ and then process the edges in this order. Each
edge that does not produce a cycle is added to $A$, edges that would
otherwise produce cycles are discarded and the procedure continues with the
next edge in the permutation.

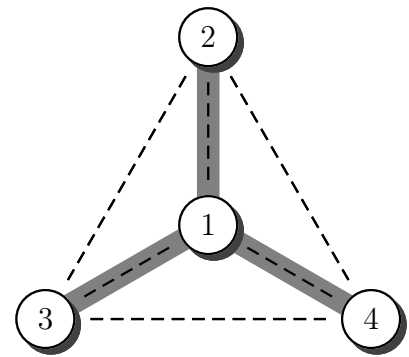
\begin{wrapfigure}{r}{0pt}
  \centering
  \begin{pspicture}(-2.8,-1.8)(2.8,3.0)
     \SpecialCoor
     \psset{shadow=true}
     \psset{radius=4mm}
     \psset{linecolor=black}
     \psset{linewidth=0.3mm}
     \psset{linestyle=solid}
     \Cnode(0;0){1}
     \Cnode(2.5;90){2}
     \Cnode(2.5;210){3}
     \Cnode(2.5;330){4}
     \psset{nodesep=2mm}
     \psset{shadow=false}
     \psset{linestyle=dashed}
     \ncline{-}{1}{2}
     \ncline{-}{1}{3}
     \ncline{-}{1}{4}
     \ncline{-}{2}{3}
     \ncline{-}{2}{4}
     \ncline{-}{3}{4}
     \psset{linestyle=solid}
    \psset{linecolor=gray}
    \psset{linewidth=3mm}
    \psset{linestyle=solid}
    \psset{nodesep=-2mm}
    \ncline{-}{1}{2}
    \ncline{-}{1}{3}
    \ncline{-}{1}{4}
     \psset{shadow=true}
     \psset{radius=4mm}
     \psset{linecolor=black}
     \psset{linewidth=0.3mm}
     \psset{linestyle=solid}
     \Cnode(0;0){1}
     \Cnode(2.5;90){2}
     \Cnode(2.5;210){3}
     \Cnode(2.5;330){4}
     \rput(0;0){1}
     \rput(2.5;90){2}
     \rput(2.5;210){3}
     \rput(2.5;330){4}
     \psset{nodesep=2mm}
     \psset{shadow=false}
     \psset{linestyle=dashed}
     \ncline{-}{1}{2}
     \ncline{-}{1}{3}
     \ncline{-}{1}{4}
     \ncline{-}{2}{3}
     \ncline{-}{2}{4}
     \ncline{-}{3}{4}
     \psset{linestyle=solid}

  \end{pspicture}
  \caption{ A star graph on $K_4$, centered at $1$.}
  \label{fig:k4}
\end{wrapfigure}
Consider the complete graph on 4 vertexes, $K_4$, and focus on the
probability of generating a star graph, centered at the vertex labeled
$1$. Figure~\ref{fig:k4} illustrates the star graph. The $K_4$ graph has
$6$ edges, hence there are $6! = 720$ different permutations. To produce
the star graph, from one such permutation, it is necessary that the edges
$(1,2)$ and $(1,3)$ are selected before the edge $(2,3)$ appears, in
general the edges $(1,u)$ and $(1,v)$ must occur before $(u,v)$. One
permutation that generates the star graph is
$(1,2), (1,3), (2,3), (1,4), (2,4),(3,4)$. Now $(2,3)$ can be moved to the
right to any of $3$ different locations so we know $4$ sequences that
generate the star graph. The same reasoning can be applied to $(2,4)$ which
can be moved once to the right. In total we counted $8$ different sequences
that generate the star graph, centered at $1$. For each of these sequences
it is possible to permute the vertexes $2, 3, 4$, amongst themselves. Hence
multiplying the previous count by $3! = 6$. In total we counted
$48 = 8 \times 6$ sequences that generate the star graph, therefore the
total probability of obtaining a star graph is $48/6!=1/15$. According to
Cayley's formula the probability to obtain the star graph centered
at $1$ should be $1/4^2=1/16$. Hence too many sequences are generating the
star graph centered at $1$.

In the next section we fix this bias by discarding some edge in the
potential cycle, not necessarily the edge that creates it.

\section{Main Idea}
\label{sec:idea}
To generate a uniform spanning tree start by generating an arbitrary
spanning tree $A$. One way to obtain this tree is to compute a depth first
search in $G$, in which case the necessary time is $O(V+E)$. In general we
wish that the mixing time of our chain is much smaller than $O(E)$,
specially for dense graphs. This initial tree is only generated once,
subsequent trees are obtained by the randomizing process. To randomize $A$
repeat the next process several times. Choose and edge $(u,v)$ from $E$,
uniformly at random, and consider the set $A \cup \{(u,v)\}$. If $(u,v)$
already belongs to $A$ the process stops, otherwise $A \cup \{(u,v)\}$
contains a cycle $C$. To complete the process choose and edge $(u',v')$
uniformly from $C \setminus \{(u,v)\}$ and remove it. Hence at each step
the set $A$ is transformed into the set
$(A \cup \{(u,v)\}) \setminus \{(u',v')\}$. An illustration of this process
is shown in Figure~\ref{fig:example}.

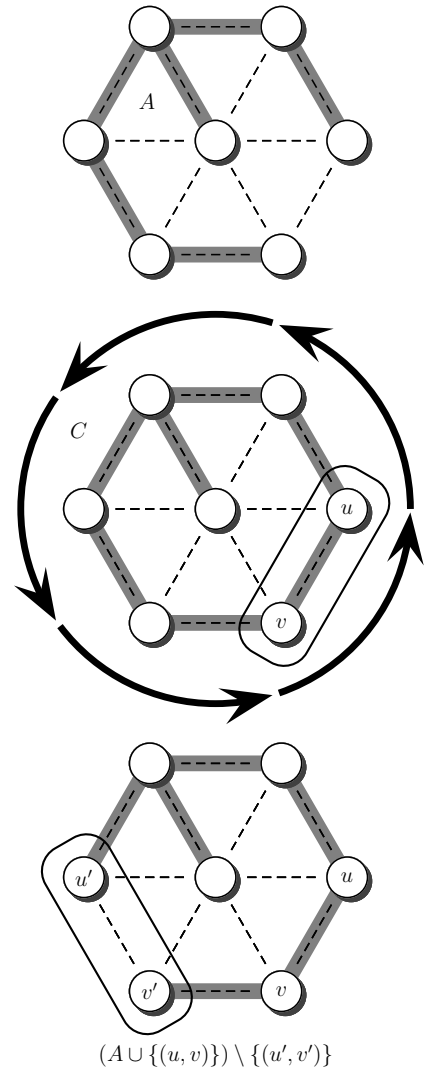
\begin{wrapfigure}{r}{0pt}
  \centering
  \scalebox{0.7}{
  \begin{pspicture}(-3.5,-17.5)(4.0,2.5)
    \rput(0,0.0){
    \baseE
    \ncline{-}{e}{f}
    \baseh
    \rput(1.5;150){$A$}
}
    \rput(0,-7.0){
      \baseE
      \ncline{-}{e}{f}
      \ncline{-}{g}{b}
      \baseh
      \psset{linecolor=black}
      \psset{linewidth=0.4mm}
      \ncbox[nodesep=3.5mm,boxsize=.7,linearc=.5]{g}{b}
      \baseh
      \rput(2.5;0){$u$}
      \rput(2.5;300){$v$}
      \psset{linewidth=1.3mm}
      \psset{arrowsize=7mm}
      \psarc[arrows=->](0,0){3.7}{0}{71}
      \psarc[arrows=->](0,0){3.7}{73}{143}
      \psarc[arrows=->](0,0){3.7}{145}{215}
      \psarc[arrows=->](0,0){3.7}{217}{287}
      \psarc[arrows=->](0,0){3.7}{289}{359}
      \rput(3.0;150){$C$}
    }
    \rput(0,-14.0){
      \baseE
      \ncline{-}{g}{b}
      \baseh
      \psset{linecolor=black}
     \psset{linewidth=0.4mm}
     \ncbox[nodesep=3.5mm,boxsize=.7,linearc=.5]{e}{f}
     \baseh
     \rput(2.5;0){$u$}
     \rput(2.5;300){$v$}
     \rput(2.5;180){$u'$}
     \rput(2.5;240){$v'$}
     \rput[B](3.5;-90){$(A \cup \{(u,v)\}) \setminus  \{(u',v')\}$}
    }
  \end{pspicture}
}
\caption{Edge swap procedure, inserting the edge $(u,v)$ into the initial
  tree $A$ generates a cycle $C$. The edge $(u',v')$ is removed from $C$.}
  \label{fig:example}
\end{wrapfigure}
This edge swapping process can be adequately modeled by a Markov chain,
where the states corresponds to different spanning trees and the
transitions among states correspond to the process we have just
described. In Section~\ref{sec:analysis} we study the ergodic properties of
this chain. For now let us focus on which data structures can be used to
compute the transition procedure efficiently. A simple solution to this
problem would be to compute a depth first search (DFS) on $A$, starting at
$u$ and terminating whenever $v$ was reached. This would allow us to
identify $C$ in $O(V)$ time, recall that $A$ contains exactly $V-1$
elements. The edge $(u',v')$ could then be easily removed. Besides $G$ the
elements of $A$ would also need to be represented with the adjacency list
data structure. For our purposes this approach is inefficient. This
computation is central to our algorithm and its complexity becomes a factor
in the overall performance. Hence we will now explain how to perform this
operation in only $O(\log V)$ amortized time with the link cut tree data
structure.

The link cut tree (LCT) is a data structure that can used to represent a
forest of rooted trees. The representation is dynamic so that edges can be
removed and added. Whenever an edge is removed the original tree is cut in
two. Adding an edge between two trees links them. This structure was
proposed by~\citet*{Sleator:1985:SBS:3828.3835}. Both the link and cut
operations can be computed in $O(\log V)$ amortized time.

The LCT can only represent trees, therefore the edge swap procedure must
first cut the edge $(u',v')$ and afterwards insert the edge $(u,v)$ with
the \texttt{Link} operation. The randomizing process needs to identify $C$
and select $(u',v')$ from it. The LCT can also compute this process in
$O(\log V)$ amortized time. The LCT works by partitioning the represented
tree into disjoint paths. Each path is stored in an auxiliary data
structure, so that any of its edges can be accessed efficiently in
$O(\log V)$ amortized time. To compute this process we force the path
$D = C \setminus\{(u,v)\}$ to become a disjoint path. This means that $D$
will be completely stored in one auxiliar data structure. Hence it is
possible to efficiently select and edge from it. Moreover the size of
$D$ can also be computed efficiently. The exact process, to force $D$ into
an auxiliar structure, is to make $u$ the root of the represented tree and
then access $v$. Algorithm~\ref{randomizeStep} shows the pseudo-code of the
edge swapping procedure. We can confirm, by inspection, that this process
can be computed in the $O(\log V)$ amortized time bound that is crucial for
our main result.

\begin{algorithm}[tbp]
  \caption{Edge swapping process}\label{randomizeStep}
  \begin{algorithmic}[1]
    \Procedure{EdgeSwap}{$A$}\Comment{$A$ is an LCT representation of the
      current spanning tree}
    \State $(u,v) \gets$ Chosen uniformly from $E$
    \If{$(u,v) \notin A$} \Comment{$O(\log V)$ time}
    \State $\mathtt{ReRoot}(A, u)$ \Comment{Makes $u$ the root of $A$}
    \State $D \gets \mathtt{Access}(A, v)$ \Comment{Obtains a
      representation of the path $C \setminus \{(u,v)\}$}
    \State $i \gets$ Chosen uniformly from $\{1, \ldots, |D|\}$
    \State $(u',v') \gets \mathtt{Select}(D, i)$ \Comment{Obtain the $i$-th
    edge from $D$.}
  \State $\mathtt{Cut}(A,u',v')$
  \State $\mathtt{Link}(A,u,v)$
    \EndIf
    \EndProcedure
  \end{algorithmic}
\end{algorithm}
\begin{theorem}
  \label{teo1}
  If $G$ is a graph and $A$ is a spanning tree of $G$ then a spanning tree
  $A'$ can be chosen uniformly from all spanning trees of $G$ in
  $O((V + \tau) \log V)$ time, where $\tau$ is the
  mixing time of an ergodic edge swapping Markov chain.
\end{theorem}
In section~\ref{sec:analysis} we prove that the process we described is
indeed an ergodic Markov chain, thus establishing the result. We finish
this section by pointing out a detail in Algorithm~\ref{randomizeStep}. In
the comment of line 3 we point out that the property $(u,v) \notin A$ must
be checked in at most $O(\log V)$ time. This can be achieved in $O(1)$ time
by keeping an array of booleans indexed by $E$. Moreover it can also be
achieved in $O(\log V)$ amortized time by using the LCT data structure,
essentially by delaying the verification until $D$ is determined and
verifying if $|D| \neq 1$.
\section{The details}
\label{sec:details}
%
\subsection{Ergodic Analysis}
\label{sec:analysis}
In this section, we analyse the Markov chain $M_t$ induced by the edge swapping
process. It should be clear that this process has the Markov property
because the probability of reaching a state depends only on the
previous state. In other words the next spanning tree depends only on the
current tree.

To prove that our procedure is correct we must show that the stationary
distribution is uniform for all states. Let us first establish that such a
stationary distribution exists. Note that, for a given finite graph $G$,
the number of spanning trees is also finite. More precisely for complete
graphs Cayley's formula yields $V^{V-2}$ spanning trees. This value is
an upper bound for other graphs, as all spanning trees of a certain graph
are also spanning trees of the complete graph with the same number of
vertexes. Therefore the chain is finite. If we show that it is irreducible
and aperiodic, it follows that it is ergodic~\cite[Corollary
7.6]{Mitzenmacher:2005:PCR:1076315} and therefore it has a stationary
distribution~\cite[Theorem 7.7]{Mitzenmacher:2005:PCR:1076315}.

The chain is aperiodic because self-loops may occur, i.e., transitions
where the underlying state does not change. Such transitions occur when
$(u,v)$ is already in $A$, therefore their probability is at least
$(V-1)/E$, because there are $V-1$ edges in a spanning tree $A$.

To establish that the chain is irreducible it is enough to show that for
any pair of states $i$ and $j$ there is a non-zero probability path from
$i$ to $j$. First note that the probability of any transition on the chain
is at least $1/(EV)$, because $(u,v)$ is chosen out of $E$ elements and
$(u',v')$ is chosen from $C \setminus \{(u,v)\}$, that contains at most
$V-1$ edges. To obtain a path from $i$ to $j$ let $A_i$ and $A_j$ represent
the respective trees. We consider the following cases:
\begin{itemize}
\item If $i=j$ we use a self-loop transition.
\item Otherwise, when $i\neq j$, it is possible to choose $(u,v)$
  from $A_j \setminus A_i$, and $(u',v')$ from
  $(C \setminus \{(u,v)\} ) \cap (A_i \setminus A_j) = C \setminus A_j $;
  note that the set equality follows from the assumption that $(u,v)$
  belongs to $A_j$. For the last property note that if no such $(u',v')$
  existed then $C \subseteq A_j$, which is a contradiction because $A_j$ is
  a tree and $C$ is a cycle. As mentioned above, the probability of this
  transition is at least $1/(EV)$. After this step the resulting tree is
  not necessarily $A_j$, but it is closer to that tree. More precisely
  $(A_i \cup \{(u,v)\}) \setminus \{(u',v')\}$ is not necessarily $A_j$,
  however the set
  $A_j \setminus ((A_i \cup \{(u,v)\}) \setminus \{(u',v')\})$ is smaller
  than the original $A_j \setminus A_i$. Its size decreases by $1$ because
  the edge $(u,v)$ exists on the second set but not on the first. Therefore
  this process can be iterated until the resulting set is empty and
  therefore the resulting tree coincides with $A_j$. The maximal size of
  $A_j \setminus A_i$ is $V-1$, because the size of $A_j$ is at most
  $V-1$. This value occurs when $A_i$ and $A_j$ do no share
  edges. Multiplying all the probabilities in the process of transforming
  $A_i$ into $A_j$ we obtain a total probability of at least
  $1/(EV)^{V-1}$.
\end{itemize}

Now that the stationary distribution is guaranteed to exist, we will show
that it coincides with the uniform distribution by proving that the chain
is time reversible~\cite[Theorem 7.10]{Mitzenmacher:2005:PCR:1076315}. We
prove that for any pair of states $i$ and $j$, with $j\neq i$, for which there exists
a transition from $i$ to $j$, with probability $P_{i,j}$, there exists,
necessarily, a transition from $j$ to $i$ with probability
$P_{j,i} = P_{i,j}$. If the transition from $i$ to $j$ exists it means that
there are edges $(u,v)$ and $(u',v')$ such that
$(A_i \cup \{(u,v)\}) \setminus \{(u',v')\} = A_j$, where $(u',v')$ belongs
to the cycle $C$ contained in $A_i \cup \{(u,v)\}$. Hence we also have that
$(A_j \cup \{(u',v')\}) \setminus \{(u,v)\} = A_i$, which means that the
tree $A_i$ can be obtained from the tree $A_j$ by adding the edge $(u',v')$
and removing the edge $(u,v)$. In other words, the process in
Figure~\ref{fig:example} is similar both top down or bottom up. This
process is a valid transition in the edge-swap chain, where the cycle $C$
is the same in both transitions, i.e., $C$ is the cycle contained in
$A_i \cup \{(u,v)\}$ and in $A_j \cup \{(u',v')\}$. Now we obtain our
result by observing that $P_{i,j} = 1/(E (C-1)) = P_{j,i}$. In the
transition from $i$ to $j$ the factor $1/E$ comes from the choice of
$(u,v)$ and the factor $1/(C-1)$ from the choice of $(u',v')$. In the
transition between $j$ to $i$, the factor $1/E$ comes from the choice of
$(u',v')$ and the factor $1/(C-1)$ from the choice of $(u,v)$. Hence we
established that the algorithm we propose correctly generates spanning
trees uniformly, provided we can sample from the stationary
distribution. Hence, we need to determine the mixing time of the chain,
i.e., the number of edge swap operations that need to be performed on an
initial tree until the distribution of the resulting trees is close enough
to the stationary distribution.

Before analyzing the mixing time of this chain we point out that it is
possible to use a faster version of this chain by choosing $(u,v)$
uniformly from $E \setminus A$, instead of from $E$. This makes the chain
faster but proving that it is aperiodic is trickier. In this chain we have
that $\Pr(M_{t+1} = i | M_{t} = i ) = 0$, for any state $i$. We will now
prove that $\Pr(M_{t+s} = i | M_{t} = i ) \neq 0$, for any state $i$ and
$s>1$. It is enough to show for $s=2$ and $s=3$, all other values follow
from the fact that the greatest common divisor of $2$ and $3$ is $1$. For
the case of $s=2$ we use the time reverse property and the following
deduction:
$\Pr(M_{t+2} = i | M_{t} = i ) \geq P_{i,j} P_{j,i} \geq (1/EV)^{2} >
0$. For the case of $s=3$ we observe that the cycle $C$ must contain at
least 3 edges $(u, v)$, $(u', v')$ and $(u'', v'')$. To obtain $A_j$ we
insert $(u,v)$ and remove $(u', v')$, now we move from this state to state
$A_k$ by inserting $(u'', v'')$ and removing $(u,v)$. Finally we move back
to $A_i$ by inserting $(u',v')$ and removing $(u'', v'')$. Hence, for this
case we have $\Pr(M_{t+3} = i | M_{t} = i) \geq (1/EV)^{3} > 0$

\subsection{A Coupling}
\label{sec:coping}

In this section, we focus on bounding the mixing time. We did not obtain
general analytical bounds from existing analysis techniques, such as
couplings~\citep{levin2017markov,Mitzenmacher:2005:PCR:1076315}, strong
stopping times~\citep{levin2017markov} and canonical
paths~\citep{sinclair1992improved}. The coupling technique yielded a bound
only for cycle graphs and moreover a simulation of the resulting coupling
converges for ladder graphs.

Before diving into the reasoning in this section, we first need a finer
understanding of the cycles generated in our process. We consider a closed
walk to be a sequence of vertexes $v_0, \ldots, v_n = v_0$, starting and
ending at the same vertex, and such that any two consecutive vertexes
$v_i$ and $v_{i+1}$ are adjacent, in our case
$(v_i,v_{i+1}) \in A \cup \{(u,v)\}$. The cycles we consider are simple,
in the sense that they consist of a set of edges for which a closed walk
can be formed, that traverses all the edges in the cycle and moreover no
vertex repetitions are allowed, except for the vertex $v_0$, which is only
repeated at the end. Formally this can be stated as: if $0 \leq i,j < n$
and $i \neq j$ then $v_i \neq v_j$.

The cycles that occur in our randomizing process are even more regular. A
cordless cycle in a graph is a cycle such that no two vertices of the
cycle are connected by an edge that does not itself belong to the
cycle. The cycles we produce also have this property, otherwise if such a
chord existed then it would form a cycle on our tree $A$, which is a
contradiction. In fact a spanning tree over a graph can alternatively be
defined as a set of edges such that for any pair of vertexes $v$ and $v'$
there is exactly one path linking $v$ to $v'$.

A coupling is an association between two copies of the same Markov chain
$X_t$ and $Y_t$, in our case the edge swapping chain. The goal of a
coupling is to make the two chains meet as fast as possible, i.e., obtain
$X_\tau=Y_\tau$, for a small value of $\tau$. At this point we say that the
chains have coalesced. The two chains may share information and cooperate
towards this goal. However, when analysed in isolation, each chain must be
indistinguishable from the original chain $M_t$. Obtaining $X_\tau=Y_\tau$
with a high probability implies that at time $\tau$ the chain is well
mixed. Precise statements of these claims are given in
Section~\ref{sec:related-work}.

We use the random variable $X_t$ to represent the state of the first chain,
at time $t$. The variable $Y_t$ represents the state of the second
chain. We consider the chain $X_t$ in state $x$ and the chain $Y_t$ in
state $y$. In one step the chain $X_t$ will transition to the state
$x' = X_{t+1}$ and the chain $Y_t$ will transition to state $y' = Y_{t+1}$.

The set $A_x \setminus A_y$ contains the edges that are exclusive to $A_x$
and likewise the set $A_y \setminus A_x$ contains the edges that are
exclusive to $A_y$. The number of such edges provides a distance
$d(x,y) = |A_x \setminus A_y| = |A_y \setminus A_x|$, that measures how far
apart are the two states. We refer to this distance as the edge
distance. We define a coupling when $d(x,y) \leq 1$, which can be extended
for states $x$ and $y$ that are farther apart, by using the path coupling
technique~\citep*{646111}.

To use the path coupling technique we cannot alter the behavior of the
chain $X_t$ as, in general, it is determined by the previous element in the
path. We denote by $i_x$ the edge that gets added to $A_x$, and by $o_x$
the edge that gets removed from the corresponding cycle
$C_x \subseteq A_x \cup \{i_x\}$, in case such a cycle exists. Likewise,
$i_y$ represents the edge that is inserted into $A_y$ and $o_y$ the edge
that gets removed from the corresponding cycle
$C_y \subseteq A_y \cup \{i_y\}$, in case such a cycle exists. The edge
$i_x$ is chosen uniformly at random from $E$ and $o_x$ is chosen uniformly
at random from $C_x \setminus \{i_x\}$. The edges $i_y$ and $o_y$ will be
obtained by trying to mimic the chain $X_t$, but still exhibiting the same
behavior as $M_t$. In this sense the information flows from $X_t$ to $Y_t$.
Let us now analyse $d(x,y)$.

\subsubsection{$d(x,y) = 0$}
If $d(x,y) = 0$ then $x = y$, which means that the corresponding trees are
also equal, $A_y = A_x$. In this case $Y_t$ uses the same transition as
$X_t$, by inserting $i_x$, i.e., set $i_y = i_x$, and removing $o_x$, i.e.,
set $o_y = o_x$.

\subsubsection{$d(x,y) = 1$}
If $d(x,y) = 1$ then the edges $e_x \in A_x \setminus A_y$ and
$e_y \in A_y \setminus A_x$ exist and are distinct. We also need the
following sets: $I= C_x \cap C_y$, \mbox{$E_x = C_x \setminus I$} and
\mbox{$E_y =C_y \setminus I$}. The set $I$ represents the edges that are
common to $C_x$ and $C_y$. The set $E_x$ represents the edges that are
exclusive to $C_x$, from the cycle point of view. This should not be
confused with $e_x$ which represents the edge that is exclusive to $A_x$,
i.e., from a tree point of view. Likewise, $E_y$ represents the edges that
are exclusive to $C_y$. Also we consider the cycle $C_e$ as the cycle
contained in $A_x \cup \{e_y\}$, which necessarily contains $e_x$. The
following Lemma describes the precise structure of these sets.
\begin{lemma}
\label{lemma:CePartition}
When $i_y = i_x$ we either have $C_x = C_y = I$ and therefore
$E_x = E_y = \emptyset$ or $E_x$, $E_y$ and $I$ form simple paths and the
following properties hold:
  \begin{itemize}
  \item $e_x \in E_x$, $e_y \in E_y$, $i_x \in I$
  \item $E_x \cap E_y = \emptyset$, $E_x \cap I = \emptyset$,
    $E_y \cap I = \emptyset$
  \item $E_x \cup I = C_x$, $E_y \cup I = C_y$, $E_x \cup E_y = C_e$.
  \end{itemize}
\end{lemma}
Notice that in particular this means that, in the non-trivial case, $E_x$
and $E_y$ partition $C_e$. A schematic representation of this Lemma is
shown in Figure~\ref{fig:schema}.
\begin{figure}[tbp]
  \input{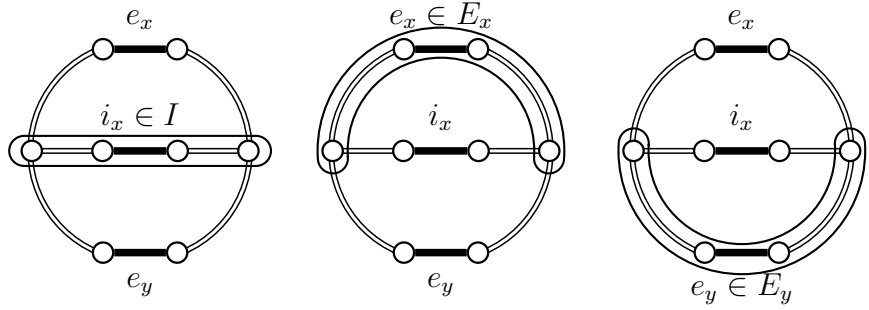}
  \renewcommand{\baseh}{

    \pscircle[style=generic]{1.5}

    \psarc[linecolor=white,linewidth=0.2](0,0){1.5}{70}{110}
    \Cnode[style=node](1.44;70){q1r}
    \Cnode[style=node](1.44;110){q1l}

    \psarc[linecolor=white,linewidth=0.2](0,0){1.5}{-110}{-70}
    \Cnode[style=node](1.44;-110){bl}
    \Cnode[style=node](1.44;-70){br}

    \psline[style=generic](0.5,0)(1.5;0)
    \psline[style=generic](-0.5,0)(1.5;180)

    \Cnode[style=node](0.5,0){cr}
    \Cnode[style=node](-0.5,0){cl}

    \Cnode[style=node](1.44;0){ir}
    \Cnode[style=node](1.44;180){il}
  }

  \begin{center}
    \begin{pspicture}[showgrid=false](-2.0,-2.0)(10.0,2.0) \SpecialCoor
      \rput(0,0){
      \baseh
      \ncline[style=edge]{-}{q1r}{q1l}
      \nbput{$e_x$}
      \ncline[style=edge]{-}{cl}{cr}
      \naput{$i_x \in I$}
      \ncline[style=edge]{-}{bl}{br}
      \nbput{$e_y$}
      \ncbox[nodesep=1.5mm,boxsize=.2,linearc=.2]{ir}{il}
    }
    \rput(4,0){
      \baseh
      \ncline[style=edge]{-}{q1r}{q1l}
      \nbput{$e_x \in E_x$}
      \ncline[style=edge]{-}{cl}{cr}
      \naput{$i_x$}
      \ncline[style=edge]{-}{bl}{br}
      \nbput{$e_y$}
      \ncarcbox[nodesep=1.5mm,boxsize=.2,
      linearc=.2,arcangle=90]{<->}{ir}{il}
    }
    \rput(8,0){
      \baseh
      \ncline[style=edge]{-}{q1r}{q1l}
      \nbput{$e_x$}
      \ncline[style=edge]{-}{cl}{cr}
      \naput{$i_x$}
      \ncline[style=edge]{-}{bl}{br}
      \nbput{$e_y \in E_y$}
      \ncarcbox[nodesep=1.5mm,boxsize=.2,
      linearc=.2,arcangle=90]{<->}{il}{ir}
    }
    \end{pspicture}
  \end{center}
  \caption{Schematic representation of the relations between $E_x$, $E_y$
    and $I$.}
  \label{fig:schema}
\end{figure}

We have several different cases described below.
Aside from the lucky cases \ref{item:swap}~and~\ref{item:ExICol}, we will
usually choose $i_y = i_x$, as $Y_t$ tries to copy $X_t$. Likewise, if
possible, we would like to set $o_y = o_x$. When this is not possible we
must choose $o_y \in E_y$, ideally we would choose $o_y = e_y$, but we must
be extra careful with this process to avoid loosing the behavior of
$M_t$. To maintain this behaviour, we must sometimes choose
$o_y \in E_y \setminus \{ e_y\}$. Since $X_t$ provides no information on
this type of edges we use $o_y = s_y$ chosen uniformly from this
$E_y \setminus \{ e_y\}$, i.e., select from $C_y$ but not $e_y$ nor edges
that are also in $C_x$.

There is a final twist to this choice, which makes the coupling non
Markovian, i.e., it does not verify the conditions in
Equations~\eqref{eq:2} and~\eqref{eq:3}. We can choose $o_y = e_y$ more
often than would otherwise be permissible, by keeping track of how $e_y$
was determined. If $e_y$ was obtained deterministically, for example by the
initial selection of $x$ and $y$, then this is not possible. In general
$e_y$ might be determined by the changes in $X_t$, in which case we want to
take advantage of the underlying randomness. Therefore, we keep track of the
random processes that occur. The exact information we store is a set of
edges $U_y \subseteq C_e \setminus \{e_x\}$ such that $e_y \in U_y$ and
moreover this set contains the edges that are equally likely to be
$e_y$. This information can be used to set $o_y = e_y$ when $s_y \in U_y$,
however after such an action the information on $U_y$ must be purged.

To illustrate the possible cases we use Figures~\ref{fig:swap}
to~\ref{fig:BigUn}, where the edges drawn with double lines represent a
generic path, that may contain several edges, or none at all. The precise
cases are the following:
  \begin{enumerate}
  \item If the chain $X_t$ loops ($x' = x$), because $i_x \in A_x$ then
    $Y_t$ also loops and therefore $y' = y$. The set $U_y$ does not change,
    i.e., set $U_{y'} = U_y$.
    \label{item:loop}
  \item If $i_x = e_y$ and $o_x = e_x$ then set $i_y = e_x$ and
    $o_y = e_y$. In this case the chains do not coalesce, they swap
    states because $x' = y$ and $y' = x$, see Figure~\ref{fig:swap}. Set
    $U_{y'} = C_{e} \setminus \{ e_{x'}\}$.
    \label{item:swap}
  \begin{figure}[tbp]
    \input{pststyles.tex}
    \renewcommand{\baseh}{
      \pscircle[style=generic]{1}

      \psarc[linecolor=white,linewidth=0.2](0,0){1}{70}{110}
      \Cnode[style=node](0.94;70){q1r}
      \Cnode[style=node](0.94;110){q1l}

      \psarc[linecolor=white,linewidth=0.2](0,0){1}{-110}{-70}
      \Cnode[style=node](0.94;-70){q2r}
      \Cnode[style=node](0.94;-110){q2l}
    }

    \begin{center}
      \begin{pspicture}[showgrid=false](-1.5,-6.7)(5,2.5) \SpecialCoor
        \psframe[framearc=.3,linewidth=0.05](-1.5,-2)(5.0,2.5)
        \psframe[framearc=.3,linewidth=0.05](-1.5,-6.7)(5.0,-2.5)
        \rput(0,0){
          \baseh \rput(0,2.0){$x$} \ncline[style=edge]{-}{q2l}{q2r}
          \nbput{$o_x = e_x \ \downarrow$}
          \ncline[style=insert]{-}{q1l}{q1r}
          \naput{$i_x = e_y \ \downarrow$} }
        \rput(3.5,0){%
          \baseh \rput(0,2.0){$y$} \ncline[style=edge]{-}{q1l}{q1r}
          \naput{$o_y = e_y \ \uparrow$} \ncline[style=insert]{-}{q2l}{q2r}
          \nbput{$i_y = e_x \ \uparrow$} }
        \rput(0,-5.0){
          \baseh \rput(0,2.0){$x'=y$} \ncline[style=edge]{-}{q1l}{q1r}
          \naput{$e_{x'}$} }
        \rput(3.5,-5.0){
          \baseh \rput(0,2.0){$y'=x$} \ncline[style=edge]{-}{q2l}{q2r}
          \nbput{$e_{y'}$} }

      \end{pspicture}
    \end{center}
    \caption{Case \ref{item:swap}.}
    \label{fig:swap}
  \end{figure}
  \item If $i_x = e_y$ and $o_x \neq e_x$ then set $i_y = e_x$ and
    $o_y = o_x$. In this case the chains coalesce, i.e., $x'=y'$, see
    Figure~\ref{fig:ExICol}. When the chains coalesce the edges $e_{x'}$
    and $e_{y'}$ no longer exist and the set $U_{y'}$ is no longer
    relevant.
    \label{item:ExICol}
  \begin{figure}[tbp]
    \input{pststyles}
    \renewcommand{\baseh}{
      \pscircle[style=generic]{1}

      \psarc[linecolor=white,linewidth=0.2](0,0){1}{10}{50}
      \Cnode[style=node](0.94;10){q1r}
      \Cnode[style=node](0.94;50){q1l}

      \psarc[linecolor=white,linewidth=0.2](0,0){1}{130}{170}
      \Cnode[style=node](0.94;130){q2r}
      \Cnode[style=node](0.94;170){q2l}

      \psarc[linecolor=white,linewidth=0.2](0,0){1}{-110}{-70}
      \Cnode[style=node](0.94;-110){bl}
      \Cnode[style=node](0.94;-70){br}
    }

    \begin{center}
      \begin{pspicture}[showgrid=false](-3.0,-6.2)(6.5,2.0) \SpecialCoor
        \psframe[framearc=.3,linewidth=0.05](-3.0,-2.0)(6.5,2.0)
        \psframe[framearc=.3,linewidth=0.05](-3.0,-6.2)(6.5,-2.5)
        \rput(0,0){%
          \baseh \rput(0,1.5){$x$} \ncline[style=edge]{-}{q1l}{q1r}
          \naput{$e_x$} \ncline[style=insert]{-}{q2r}{q2l}
          \nbput{$i_x = e_y \ \downarrow$} \ncline[style=edge]{-}{bl}{br}
          \nbput{$o_x \ \downarrow$} }
        \rput(3.5,0){ \baseh \rput(0,1.5){$y$}
          \ncline[style=insert]{-}{q1l}{q1r}
          \naput{$\downarrow \ i_y = e_x$} \ncline[style=edge]{-}{q2r}{q2l}
          \nbput{$e_y$} \ncline[style=edge]{-}{bl}{br}
          \nbput{$o_y = o_x \ \downarrow$} }
        \rput(1.75,-4.5){%
          \baseh \rput(0,1.5){$x' = y'$} \ncline[style=edge]{-}{q1l}{q1r}
          \ncline[style=edge]{-}{q2r}{q2l} }
      \end{pspicture}
    \end{center}
    \caption{Case \ref{item:ExICol}.}
    \label{fig:ExICol}
  \end{figure}
  \item If $i_x \neq e_y$ set $i_y = i_x$. We now have $3$ sub-cases, which
    are further sub-divided. These cases depend on whether $|C_x| = |C_y|$,
    $|C_x| < |C_y|$ or $|C_x| > |C_y|$. We start with $|C_x| = |C_y|$ which
    is simpler and establishes the basic situations. When $|C_x| < |C_y|$
    or $|C_x| > |C_y|$ we use some Bernoulli random variables to balance
    out probabilities and whenever possible reduce to the cases considered
    for $|C_x| = |C_y|$. When this is not possible we present the
    corresponding new situation.
    \label{item:GenI}
    \begin{enumerate}
    \item If $|C_x| = |C_y|$ we have the following situations:
      \label{item:Eqs}
      \begin{enumerate}
      \item If $o_x = e_x$ then set $o_y = e_y$. In this case the chains
        coalesce, see Figure~\ref{fig:EqOeX}.
        \label{item:EqOeX}
        \begin{figure}[tbp]
          \input{pststyles}
          \renewcommand{\baseh}{

            \pscircle[style=generic]{1.5}

            \psarc[linecolor=white,linewidth=0.2](0,0){1.5}{70}{110}
            \Cnode[style=node](1.44;70){q1r}
            \Cnode[style=node](1.44;110){q1l}

            \psarc[linecolor=white,linewidth=0.2](0,0){1.5}{-110}{-70}
            \Cnode[style=node](1.44;-110){bl}
            \Cnode[style=node](1.44;-70){br}

            \psline[style=generic](0.5,0)(1.5;0)
            \psline[style=generic](-0.5,0)(1.5;180)

            \Cnode[style=node](0.5,0){cr}
            \Cnode[style=node](-0.5,0){cl}

            \Cnode[style=node](1.44;0){ir}
            \Cnode[style=node](1.44;180){il}
          }

          \begin{center}
            \begin{pspicture}[showgrid=false](-2.0,-7.5)(6.0,3.0) \SpecialCoor
              \psframe[framearc=.3,linewidth=0.05](-2.0,-2.5)(6.0,3.0)
              \psframe[framearc=.3,linewidth=0.05](-2.0,-7.5)(6.0,-3.0)
              \rput(0,0){%
                \baseh \rput(0,2.5){$x$} \ncline[style=insert]{-}{q1r}{q1l}
                \nbput{$i_x \ \downarrow$} \ncline[style=edge]{-}{cl}{cr}
                \nbput{$o_x = e_x \ \downarrow$} }
              \rput(4.0,0){ \baseh \rput(0,2.5){$y$}
                \ncline[style=insert]{-}{q1r}{q1l}
                \nbput{$i_y = i_x \ \downarrow$} \ncline[style=edge]{-}{bl}{br}
                \nbput{$o_y = e_y \ \downarrow$} }
              \rput(2.0,-5.5){%
                \baseh \rput(0,2.0){$x' = y'$} \ncline[style=edge]{-}{q1r}{q1l} }
            \end{pspicture}
          \end{center}
          \caption{Case~\ref{item:EqOeX}, case~\ref{item:SmUnC} and
            case~\ref{item:BigEx}.}
          \label{fig:EqOeX}
        \end{figure}
      \item If $o_x \in I \setminus \{i_x\}$ then set $o_y = o_x$. In this
        case the chains do not coalesce, in fact the exclusive edges remain
        unchanged, i.e., $e_{x'} = e_x$ and $e_{y'} = e_y $, see
        Figures~\ref{fig:EqCommonNix}~and~\ref{fig:EqCommonInx}. When
        $o_x \notin C_e$ the set $C_{e'}$ remains equal to $C_e$ and
        likewise $U_{y'}$ remains equal to $U_y$, see
        Figure~\ref{fig:EqCommonNix}. Otherwise when $o_x \in C_e$ the set
        $C_{e'}$ is different from $C_e$ and we assign
        $U_{y'} = U_y \cap C_{e'}$, see Figure~\ref{fig:EqCommonInx}.
        \label{item:EqCommon}
        \begin{figure}[tbp]
          \input{pststyles}
          \renewcommand{\baseh}{

            \pscircle[style=generic]{1}

            \psarc[linecolor=white,linewidth=0.2](0,0){1}{10}{50}
            \Cnode[style=node](0.94;10){q1r}
            \Cnode[style=node](0.94;50){q1l}

            \psarc[linecolor=white,linewidth=0.2](0,0){1}{130}{170}
            \Cnode[style=node](0.94;130){q2r}
            \Cnode[style=node](0.94;170){q2l}

            \psarc[linecolor=white,linewidth=0.2](0,0){1}{-110}{-70}
            \Cnode[style=node](0.94;-110){bl}
            \Cnode[style=node](0.94;-70){br}

            \psline[style=generic](0.35,-0.3)(1.0;-20)
            \psline[style=generic](-0.35,-0.3)(1.0;-160)

            \Cnode[style=node](0.35,-0.3){cr}
            \Cnode[style=node](-0.35,-0.3){cl}

            \Cnode[style=node](0.94;-20){ir}
            \Cnode[style=node](0.94;-160){il}
          }

          \begin{center}
            \begin{pspicture}[showgrid=false](-2.0,-6.2)(8.5,2.0) \SpecialCoor
              \psframe[framearc=.3,linewidth=0.05](-2.0,-2.0)(8.5,2.0)
              \psframe[framearc=.3,linewidth=0.05](-2.0,-6.2)(8.5,-2.5)
              \rput(0,0){%
                \baseh \rput(0,1.5){$x$} \ncline[style=edge]{-}{q1l}{q1r}
                \naput{$\uparrow \ o_x$} \ncline[style=insert]{-}{q2r}{q2l}
                \nbput{$i_x \ \downarrow$} \ncline[style=edge]{-}{bl}{br}
                \nbput{$e_x$} }
              \rput(5.3,0){ \baseh \rput(0,1.5){$y$}
                \ncline[style=edge]{-}{q1l}{q1r} \naput{$\uparrow \ o_y = o_x$}
                \ncline[style=insert]{-}{q2r}{q2l}
                \nbput{$i_y = i_x \ \downarrow$} \ncline[style=edge]{-}{cl}{cr}
                \naput{$e_y$} }
              \rput(0,-4.5){%
                \baseh \rput(0,1.5){$x'$} \ncline[style=edge]{-}{q2r}{q2l}
                \ncline[style=edge]{-}{bl}{br} \nbput{$e_{x'} = e_x$} }
              \rput(5.3,-4.5){%
                \baseh \rput(0,1.5){$y'$} \ncline[style=edge]{-}{q2r}{q2l}
                \ncline[style=edge]{-}{cl}{cr} \naput{$e_{y'}=e_y$} }
            \end{pspicture}
          \end{center}
          \caption{Case~\ref{item:EqCommon}, case~\ref{item:SmCommonUy} and
            case~\ref{item:BigCommon}, when $o_x \notin C_e$.}
          \label{fig:EqCommonNix}
        \end{figure}
        \begin{figure}[tbp]
          \input{pststyles}
          \renewcommand{\baseh}{

            \pscircle[style=generic]{1.5}

            \psarc[linecolor=white,linewidth=0.2](0,0){1.5}{10}{50}
            \Cnode[style=node](1.44;10){q1r}
            \Cnode[style=node](1.44;50){q1l}

            \psarc[linecolor=white,linewidth=0.2](0,0){1.5}{130}{170}
            \Cnode[style=node](1.44;130){q2r}
            \Cnode[style=node](1.44;170){q2l}

            \psarc[linecolor=white,linewidth=0.2](0,0){1.5}{-110}{-70}
            \Cnode[style=node](1.44;-110){bl}
            \Cnode[style=node](1.44;-70){br}

            \psline[style=generic](0.5,-0.3)(1.5;-20)
            \psline[style=generic](-0.5,-0.3)(1.5;-160)

            \Cnode[style=node](0.5,-0.3){cr}
            \Cnode[style=node](-0.5,-0.3){cl}

            \Cnode[style=node](1.44;-20){ir}
            \Cnode[style=node](1.44;-160){il}
          }

          \begin{center}
            \begin{pspicture}[showgrid=false](-2.0,-7.5)(8.5,2.5) \SpecialCoor
              \psframe[framearc=.3,linewidth=0.05](-2.0,-2.5)(8.5,2.5)
              \psframe[framearc=.3,linewidth=0.05](-2.0,-7.5)(8.5,-3.0)
              \rput(0,0){%
                \baseh \rput(0,2.0){$x$} \ncline[style=edge]{-}{bl}{br}
                \nbput{$\downarrow \ o_x$} \ncline[style=insert]{-}{cr}{cl}
                \nbput{$i_x \ \downarrow$} \ncline[style=edge]{-}{q1r}{q1l}
                \nbput{$e_x$} }
              \rput(6.5,0){ \baseh \rput(0,2.0){$y$}
                \ncline[style=edge]{-}{bl}{br} \nbput{$\downarrow \ o_y = o_x$}
                \ncline[style=insert]{-}{cr}{cl} \nbput{$i_y = i_x \ \downarrow$}
                \ncline[style=edge]{-}{q2r}{q2l} \nbput{$e_y$} }
              \rput(0,-5.5){%
                \baseh \rput(0,2.0){$x'$} \ncline[style=edge]{-}{cr}{cl}
                \ncline[style=edge]{-}{q1r}{q1l} \nbput{$e_{x'} = e_x$} }
              \rput(6.5,-5.5){%
                \baseh \rput(0,2.0){$y'$} \ncline[style=edge]{-}{cr}{cl}
                \ncline[style=edge]{-}{q2r}{q2l} \nbput{$e_{y'} = e_y$} }
            \end{pspicture}
          \end{center}
          \caption{Case~\ref{item:EqCommon}, case~\ref{item:SmCommonUy} and
            case~\ref{item:BigCommon}, when $o_x \in C_e$.}
          \label{fig:EqCommonInx}
        \end{figure}
      \item If $o_x \in E_x \setminus \{e_x\}$ then
        select $s_y$ uniformly from $E_y \setminus \{e_y\}$. If
        $s_y \in U_y$ then set $o_y = e_y$, see Figure~\ref{fig:EqUy}. In
        this case set $U_{y'} = E_x \setminus \{e_x\}$. The
        alternative, when $s_y \notin U_y$ is considered in the next case
        (\ref{item:EqUnC}).
        \label{item:EqUy}
        \begin{figure}[tbp]
          \input{pststyles}
          \renewcommand{\baseh}{

            \pscircle[style=generic]{1.5}

            \psarc[linecolor=white,linewidth=0.2](0,0){1.5}{10}{50}
            \Cnode[style=node](1.44;10){q1r}
            \Cnode[style=node](1.44;50){q1l}

            \psarc[linecolor=white,linewidth=0.2](0,0){1.5}{130}{170}
            \Cnode[style=node](1.44;130){q2r}
            \Cnode[style=node](1.44;170){q2l}

            \psarc[linecolor=white,linewidth=0.2](0,0){1.5}{-110}{-70}
            \Cnode[style=node](1.44;-110){bl}
            \Cnode[style=node](1.44;-70){br}

            \psline[style=generic](0.5,-0.3)(1.5;-20)
            \psline[style=generic](-0.5,-0.3)(1.5;-160)

            \Cnode[style=node](0.5,-0.3){cr}
            \Cnode[style=node](-0.5,-0.3){cl}

            \Cnode[style=node](1.44;-20){ir}
            \Cnode[style=node](1.44;-160){il}
          }

          \begin{center}
            \begin{pspicture}[showgrid=false](-2.5,-7.5)(8.0,2.5) \SpecialCoor
              \psframe[framearc=.3,linewidth=0.05](-2.5,-2.5)(8.0,2.5)
              \psframe[framearc=.3,linewidth=0.05](-2.5,-7.5)(8.0,-3.0)
              \rput(0,0){%
                \baseh \rput(0,2.0){$x$} \ncline[style=edge]{-}{q2l}{q2r}
                \naput{$o_x \ \uparrow$} \ncline[style=insert]{-}{cr}{cl}
                \nbput{$i_x \ \downarrow$} \ncline[style=edge]{-}{q1r}{q1l}
                \nbput{$e_x$} }
              \rput(5.5,0){ \baseh \rput(0,2.0){$y$}
                \ncline[style=edge]{-}{bl}{br} \nbput{$o_y = e_y \ \downarrow$}
                \ncline[style=insert]{-}{cr}{cl} \nbput{$i_y = i_x \ \downarrow$}
                \ncline[style=edge]{-}{q2r}{q2l} }
              \rput(0,-5.5){%
                \baseh \rput(0,2.0){$x'$} \ncline[style=edge]{-}{cr}{cl}
                \ncline[style=edge]{-}{q1r}{q1l} \nbput{$e_{x'}=e_x$} }
              \rput(5.5,-5.5){%
                \baseh \rput(0,2.0){$y'$} \ncline[style=edge]{-}{cr}{cl}
                \ncline[style=edge]{-}{q2r}{q2l} \nbput{$e_{y'}$} }
            \end{pspicture}
          \end{center}
          \caption{Case~\ref{item:EqUy}, case~\ref{item:SmUnR} and case
            \ref{item:BigBB}, when $B'$ is true.}
          \label{fig:EqUy}
        \end{figure}
      \item If $o_x \in E_x \setminus \{e_x\}$ and $s_y \notin U_y$, then
        set $o_y = s_y$. This case is shown in Figure~\ref{fig:EqUnC}. In
        this case the distance of the coupled states increases, i.e.,
        $d(x', y') = 2$. Therefore we include a new state $z'$, in between
        $x'$ and $y'$ and define $e_{z'}$ to be the edge in
        $A_{z'} \setminus A_{x'}$; and $e_{y'}$ the edge in
        $A_{y'} \setminus A_{z'}$; and $e_{x'}$ the edge in
        $A_{x'} \setminus A_{z'}$. The set $U_{z'}$ should contain the
        edges that provide alternatives to $e_{z'}$. In this case set
        $U_{z'} = E_x \setminus \{e_x\}$ and
        $U_{y'} = (U_y \cap E_y) \setminus \{o_y\}$.
        \label{item:EqUnC}
        \begin{figure}[tbp]
          \input{pststyles}
          \renewcommand{\baseh}{

            \pscircle[style=generic]{1.5}

            \psarc[linecolor=white,linewidth=0.2](0,0){1.5}{36}{72}
            \Cnode[style=node](1.44;36){q1r}
            \Cnode[style=node](1.44;72){q1l}

            \psarc[linecolor=white,linewidth=0.2](0,0){1.5}{108}{144}
            \Cnode[style=node](1.44;108){q2r}
            \Cnode[style=node](1.44;144){q2l}

            \psarc[linecolor=white,linewidth=0.2](0,0){1.5}{216}{252}
            \Cnode[style=node](1.44;216){b1r}
            \Cnode[style=node](1.44;252){b1l}

            \psarc[linecolor=white,linewidth=0.2](0,0){1.5}{288}{324}
            \Cnode[style=node](1.44;288){b2r}
            \Cnode[style=node](1.44;324){b2l}

            \psline[style=generic](0.5,0)(1.5;0)
            \psline[style=generic](-0.5,0)(1.5;180)

            \Cnode[style=node](0.5,0){cr}
            \Cnode[style=node](-0.5,0){cl}

            \Cnode[style=node](1.44;0){ir}
            \Cnode[style=node](1.44;180){il}
          }

          \begin{center}
            \begin{pspicture}[showgrid=false](-2.0,-7.5)(11.5,2.5) \SpecialCoor
              \psframe[framearc=.3,linewidth=0.05](-2.0,-2.3)(11.5,2.5)
              \psframe[framearc=.3,linewidth=0.05](-2.0,-7.5)(11.5,-2.8)
              \rput(0,0){%
                \baseh \rput(0,2.0){$x$} \ncline[style=edge]{-}{q1r}{q1l}
                \nbput{$e_x$} \ncline[style=edge]{-}{q2r}{q2l}
                \nbput{$o_x \ \uparrow$} \ncline[style=insert]{-}{cl}{cr}
                \naput{$i_x \ \downarrow$} \ncline[style=edge]{-}{b2r}{b2l} }
              \rput(8.5,0){ \baseh \rput(0,2.0){$y$}
                \ncline[style=edge]{-}{q2r}{q2l} \ncline[style=edge]{-}{b1r}{b1l}
                \nbput{$e_y$} \ncline[style=edge]{-}{b2r}{b2l}
                \nbput{$\downarrow \ o_y = s_y$} \ncline[style=insert]{-}{cl}{cr}
                \naput{$i_y = i_x \ \downarrow$} }
              \rput(0,-5.5){%
                \baseh \rput(0,2.0){$x'$} \ncline[style=edge]{-}{q1r}{q1l}
                \nbput{$e_{x'}=e_x$} \ncline[style=edge]{-}{cl}{cr}
                \ncline[style=edge]{-}{b2r}{b2l} }
              \rput(4.5,-5.5){ \baseh \rput(0,2.0){$z'$}
                \ncline[style=edge]{-}{q2r}{q2l} \nbput{$e_{z'}$}
                \ncline[style=edge]{-}{b2r}{b2l} \ncline[style=edge]{-}{cl}{cr} }
              \rput(8.5,-5.5){%
                \baseh \rput(0,2.0){$y'$} \ncline[style=edge]{-}{q2r}{q2l}
                \ncline[style=edge]{-}{b1r}{b1l} \nbput{$e_{y'}=e_y$}
                \ncline[style=edge]{-}{cl}{cr} }
            \end{pspicture}
          \end{center}
          \caption{Case \ref{item:EqUnC}, case~\ref{item:SmUnR} and
            case~\ref{item:BigBB}.}
          \label{fig:EqUnC}
        \end{figure}
      \end{enumerate}
    \item If $|C_x| < |C_y|$ then $X_t$ will choose $o_x \in I$ with a
      higher probability then what $Y_t$ should. Therefore we use a
      Bernoulli random variable $B$ with a success probability $p$ defined
      as follows:
      \begin{equation*}
        p = \frac{C_x-1}{C_y-1}
      \end{equation*}
      In Lemma~\ref{lemma:Markovian} we prove that $p$ properly balances
      the necessary probabilities, for now note that when $|C_x| = |C_y|$
      the expression for $p$ yields $p = 1$. This is coherent with the
      following cases, because when $B$ yields \texttt{true} we use the
      choices defined for $|C_x| = |C_y|$. The following situations are
      possible:
      \label{item:Smx}
      \begin{enumerate}
      \item If $o_x = e_x$ then we reduce to the case~\ref{item:EqOeX},
        both when $B$ yields \texttt{true} or when $B$ fails and
        $s_y \in U_y$. Set $o_y = e_y$, see Figure~\ref{fig:EqOeX}.  The
        new case occurs when $B$ fails and $s_y \notin U_y$, in this
        situation set $o_y = s_y$ and
        $U_{y'} = (U_y \cap C_y) \setminus \{o_y\}$, see
        Figure~\ref{fig:SmUnC}.
        \label{item:SmUnC}
        \begin{figure}[tbp]
          \input{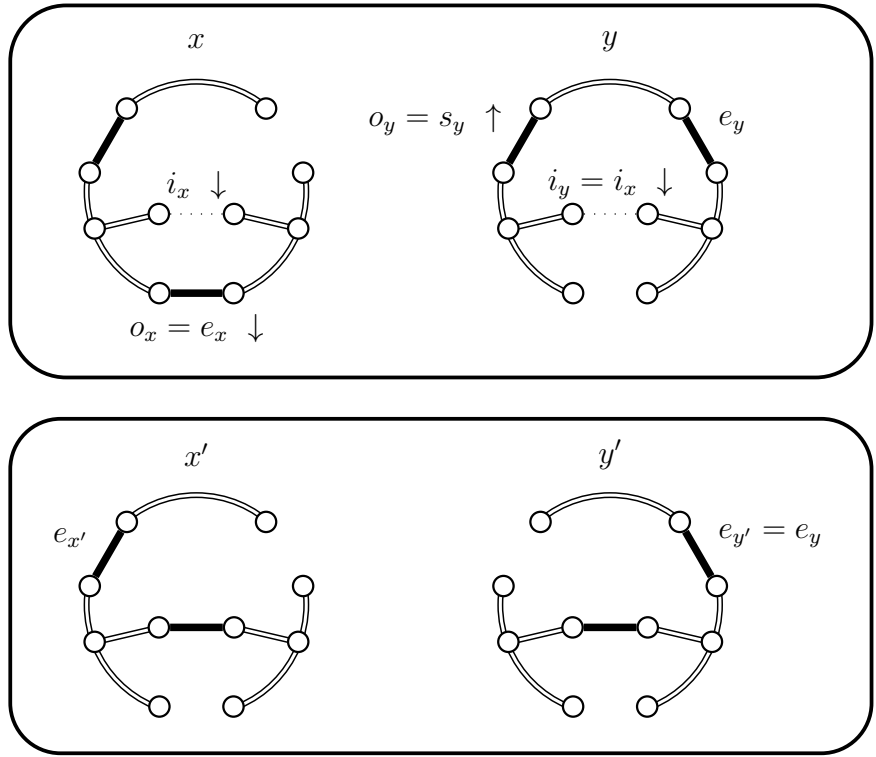}
          \renewcommand{\baseh}{

            \pscircle[style=generic]{1.5}

            \psarc[linecolor=white,linewidth=0.2](0,0){1.5}{10}{50}
            \Cnode[style=node](1.44;10){q1r}
            \Cnode[style=node](1.44;50){q1l}

            \psarc[linecolor=white,linewidth=0.2](0,0){1.5}{130}{170}
            \Cnode[style=node](1.44;130){q2r}
            \Cnode[style=node](1.44;170){q2l}

            \psarc[linecolor=white,linewidth=0.2](0,0){1.5}{-110}{-70}
            \Cnode[style=node](1.44;-110){bl}
            \Cnode[style=node](1.44;-70){br}

            \psline[style=generic](0.5,-0.3)(1.5;-20)
            \psline[style=generic](-0.5,-0.3)(1.5;-160)

            \Cnode[style=node](0.5,-0.3){cr}
            \Cnode[style=node](-0.5,-0.3){cl}

            \Cnode[style=node](1.44;-20){ir}
            \Cnode[style=node](1.44;-160){il}
          }

          \begin{center}
            \begin{pspicture}[showgrid=false](-2.5,-7.5)(9.0,2.5) \SpecialCoor
              \psframe[framearc=.3,linewidth=0.05](-2.5,-2.5)(9.0,2.5)
              \psframe[framearc=.3,linewidth=0.05](-2.5,-7.5)(9.0,-3.0)
              \rput(0,0){%
                \baseh \rput(0,2.0){$x$} \ncline[style=edge]{-}{bl}{br}
                \nbput{$o_x = e_x \ \downarrow$} \ncline[style=insert]{-}{cr}{cl}
                \nbput{$i_x \ \downarrow$} \ncline[style=edge]{-}{q2r}{q2l} }
              \rput(5.5,0){ \baseh \rput(0,2.0){$y$}
                \ncline[style=edge]{-}{q2l}{q2r} \naput{$o_y = s_y \ \uparrow$}
                \ncline[style=insert]{-}{cr}{cl} \nbput{$i_y = i_x\ \downarrow$}
                \ncline[style=edge]{-}{q1r}{q1l} \nbput{$e_y$} }
              \rput(0,-5.5){%
                \baseh \rput(0,2.0){$x'$} \ncline[style=edge]{-}{cr}{cl}
                \ncline[style=edge]{-}{q2r}{q2l} \nbput{$e_{x'}$} }
              \rput(5.5,-5.5){%
                \baseh \rput(0,2.0){$y'$} \ncline[style=edge]{-}{cr}{cl}
                \ncline[style=edge]{-}{q1r}{q1l} \nbput{$e_{y'} = e_y$} }
            \end{pspicture}
          \end{center}

          \caption{Case~\ref{item:SmUnC} when $B$ fails and $s_y \notin U_y$.}
          \label{fig:SmUnC}
        \end{figure}
      \item If $o_x \in I \setminus \{i_x\}$ then we reduce to the
        case~\ref{item:EqCommon} when $B$ yields \texttt{true}. Set
        $o_y=o_x$, see
        Figures~\ref{fig:EqCommonNix}~and~\ref{fig:EqCommonInx}. When $B$
        fails and $s_y \in U_y$ we have a new situation. Set $o_y = e_y$
        and $U_{y'} = I \setminus \{i_x\}$.  The chains
        preserve their distance, i.e., $d(x',y') = 1$, see
        Figure~\ref{fig:SmCommonUy}.
        The alternative, when $s_y \notin U_y$ is considered in the next
        case (\ref{item:SmComUncommon}).
        \label{item:SmCommonUy}
        \begin{figure}[tbp]
          \input{pststyles}
          \renewcommand{\baseh}{
            \pscircle[style=generic]{1}

            \psarc[linecolor=white,linewidth=0.2](0,0){1}{10}{50}
            \Cnode[style=node](0.94;10){q1r}
            \Cnode[style=node](0.94;50){q1l}

            \psarc[linecolor=white,linewidth=0.2](0,0){1}{130}{170}
            \Cnode[style=node](0.94;130){q2r}
            \Cnode[style=node](0.94;170){q2l}

            \psarc[linecolor=white,linewidth=0.2](0,0){1}{-110}{-70}
            \Cnode[style=node](0.94;-110){bl}
            \Cnode[style=node](0.94;-70){br}

            \psline[style=generic](0.35,-0.3)(1.0;-20)
            \psline[style=generic](-0.35,-0.3)(1.0;-160)

            \Cnode[style=node](0.35,-0.3){cr}
            \Cnode[style=node](-0.35,-0.3){cl}

            \Cnode[style=node](0.94;-20){ir}
            \Cnode[style=node](0.94;-160){il}
          }

          \begin{center}
            \begin{pspicture}[showgrid=false](-2.0,-6.2)(7.5,2.0) \SpecialCoor
              \psframe[framearc=.3,linewidth=0.05](-2.0,-2.0)(7.5,2.0)
              \psframe[framearc=.3,linewidth=0.05](-2.0,-6.2)(7.5,-2.5)
              \rput(0,0){%
                \baseh \rput(0,1.5){$x$} \ncline[style=edge]{-}{q1l}{q1r}
                \naput{$\uparrow \ o_x$} \ncline[style=insert]{-}{q2r}{q2l}
                \nbput{$i_x \ \downarrow$} \ncline[style=edge]{-}{cl}{cr}
                \naput{$e_x$} }
              \rput(5.3,0){ \baseh \rput(0,1.5){$y$}
                \ncline[style=edge]{-}{q1l}{q1r}
                \ncline[style=insert]{-}{q2r}{q2l}
                \nbput{$i_y = i_x \ \downarrow$} \ncline[style=edge]{-}{bl}{br}
                \nbput{$o_y = e_y \ \downarrow$} }
              \rput(0,-4.5){%
                \baseh \rput(0,1.5){$x'$} \ncline[style=edge]{-}{q2r}{q2l}
                \ncline[style=edge]{-}{cl}{cr} \naput{$e_{x'}=e_x$} }
              \rput(5.3,-4.5){%
                \baseh \rput(0,1.5){$y'$} \ncline[style=edge]{-}{q1l}{q1r}
                \naput{$e_{y'}$} \ncline[style=edge]{-}{q2r}{q2l} }
            \end{pspicture}
          \end{center}
          \caption{Case \ref{item:SmCommonUy}.}
          \label{fig:SmCommonUy}
        \end{figure}
      \item If $o_x \in I \setminus \{i_x\}$ and $B$ fails and
        $s_y \notin U_y$.  We have a new situation, set $o_y = s_y$.  The
        distance increases, $d(x',y') = 2$, see
        Figure~\ref{fig:SmComUncommon}. Set
        $U_{z'} = I \setminus \{i_x\}$ and
        $U_{y'} = (U_y \cap E_y) \setminus \{o_y\}$.
        \label{item:SmComUncommon}
        \begin{figure}[tbp]
          \input{pststyles}
          \renewcommand{\baseh}{
            \pscircle[style=generic]{1.5}

            \psarc[linecolor=white,linewidth=0.2](0,0){1.5}{36}{72}
            \Cnode[style=node](1.44;36){q1r}
            \Cnode[style=node](1.44;72){q1l}

            \psarc[linecolor=white,linewidth=0.2](0,0){1.5}{108}{144}
            \Cnode[style=node](1.44;108){q2r}
            \Cnode[style=node](1.44;144){q2l}

            \psarc[linecolor=white,linewidth=0.2](0,0){1.5}{216}{252}
            \Cnode[style=node](1.44;216){b1r}
            \Cnode[style=node](1.44;252){b1l}

            \psarc[linecolor=white,linewidth=0.2](0,0){1.5}{288}{324}
            \Cnode[style=node](1.44;288){b2r}
            \Cnode[style=node](1.44;324){b2l}

            \psline[style=generic](0.5,0)(1.5;0)
            \psline[style=generic](-0.5,0)(1.5;180)

            \Cnode[style=node](0.5,0){cr}
            \Cnode[style=node](-0.5,0){cl}

            \Cnode[style=node](1.44;0){ir}
            \Cnode[style=node](1.44;180){il}
          }

          \begin{center}
            \begin{pspicture}[showgrid=false](-2.0,-7.5)(11.0,2.5) \SpecialCoor
              \psframe[framearc=.3,linewidth=0.05](-2.0,-2.3)(11.0,2.5)
              \psframe[framearc=.3,linewidth=0.05](-2.0,-7.5)(11.0,-2.8)
              \rput(0,0){%
                \baseh \rput(0,2.0){$x$} \ncline[style=edge]{-}{q1r}{q1l}
                \nbput{$\uparrow \ o_x$} \ncline[style=insert]{-}{q2r}{q2l}
                \nbput{$i_x \ \downarrow$} \ncline[style=edge]{-}{cl}{cr}
                \nbput{$e_x$} \ncline[style=edge]{-}{b1r}{b1l} }
              \rput(8.0,0){ \baseh \rput(0,2.0){$y$}
                \ncline[style=edge]{-}{q1r}{q1l}
                \ncline[style=insert]{-}{q2r}{q2l}
                \nbput{$i_y = i_x \ \downarrow$} \ncline[style=edge]{-}{b2r}{b2l}
                \nbput{$e_y$} \ncline[style=edge]{-}{b1r}{b1l}
                \nbput{$o_y = s_y \ \downarrow$} }
              \rput(0,-5.5){%
                \baseh \rput(0,2.0){$x'$} \ncline[style=edge]{-}{q2r}{q2l}
                \ncline[style=edge]{-}{cl}{cr} \nbput{$e_{x'}=e_x$}
                \ncline[style=edge]{-}{b1r}{b1l} }
              \rput(4.0,-5.5){ \baseh \rput(0,2.0){$z'$}
                \ncline[style=edge]{-}{q2r}{q2l} \ncline[style=edge]{-}{q1r}{q1l}
                \nbput{$e_{z'}$} \ncline[style=edge]{-}{b1r}{b1l} }
              \rput(8.0,-5.5){%
                \baseh \rput(0,2.0){$y'$} \ncline[style=edge]{-}{q1r}{q1l}
                \ncline[style=edge]{-}{q2r}{q2l} \ncline[style=edge]{-}{b2r}{b2l}
                \nbput{$e_{y'}=e_y$} }
            \end{pspicture}
          \end{center}
          \caption{Case \ref{item:SmComUncommon}.}
          \label{fig:SmComUncommon}
        \end{figure}
      \item If $o_x \in E_x \setminus \{e_x\}$ then if $s_y \in U_y$ use
        case~\ref{item:EqUy} (Figure~\ref{fig:EqUy}), otherwise, when
        $s_y \notin U_y$, use case~\ref{item:EqUnC}
        (Figure~\ref{fig:EqUnC}).
        \label{item:SmUnR}
      \end{enumerate}
    \item If $|C_x| > |C_y|$ we have the following situations:
      \label{item:Bigx}
      \begin{enumerate}
      \item If $o_x = e_x$ then use case~\ref{item:EqOeX} and set
        $o_y = e_y$, see Figure~\ref{fig:EqOeX}. The chains coalesce.
        \label{item:BigEx}
      \item If $o_x \in I \setminus \{i_x\}$ then use
        case~\ref{item:EqCommon} and set $o_y = o_x$, see
        Figures~\ref{fig:EqCommonNix}~and~\ref{fig:EqCommonInx}.
        \label{item:BigCommon}
      \item If $o_x \in E_x \setminus \{e_x\}$ then we use a new Bernoulli
        random variable $B^*$ with a success probability $p^*$ defined as
        follows:
        \begin{equation*}
          p^* = \left( \frac{1}{C_y-1} - \frac{1}{C_x-1}\right) \times \frac{(C_x-1)(I-1)}{E_x-1}
        \end{equation*}
        In Lemma~\ref{lemma:Markovian} we prove that $B^*$ properly
        balances the necessary probabilities. For now, note that when
        $|C_x| = |C_y|$ the expression for $p^*$ yields $p^* = 0$, because
        \mbox{$1/(C_y-1) - 1/(C_x-1)$} becomes $0$. This is coherent because
        when $B^*$ returns \texttt{false} we will use the choices defined
        for $|C_x| = |C_y|$.  The case when $B^*$ fails is considered in
        the next case (\ref{item:BigBB}).

        If $B^*$ is successful we have a new situation. Set $o_y = s_i$,
        where $s_i$ is chosen uniformly from \mbox{$I \setminus \{i_y\}$}
        and see Figure~\ref{fig:BigUn}. We have $e_{y'} = e_y$,
        $U_{y'} = (U_y \cap E_y) \setminus \{o_y\}$, $e_{z'} = o_x$ and
        $U_{z'} = E_x \setminus \{e_x\}$.
        \label{item:BigUn}
        \begin{figure}[tbp]
          \input{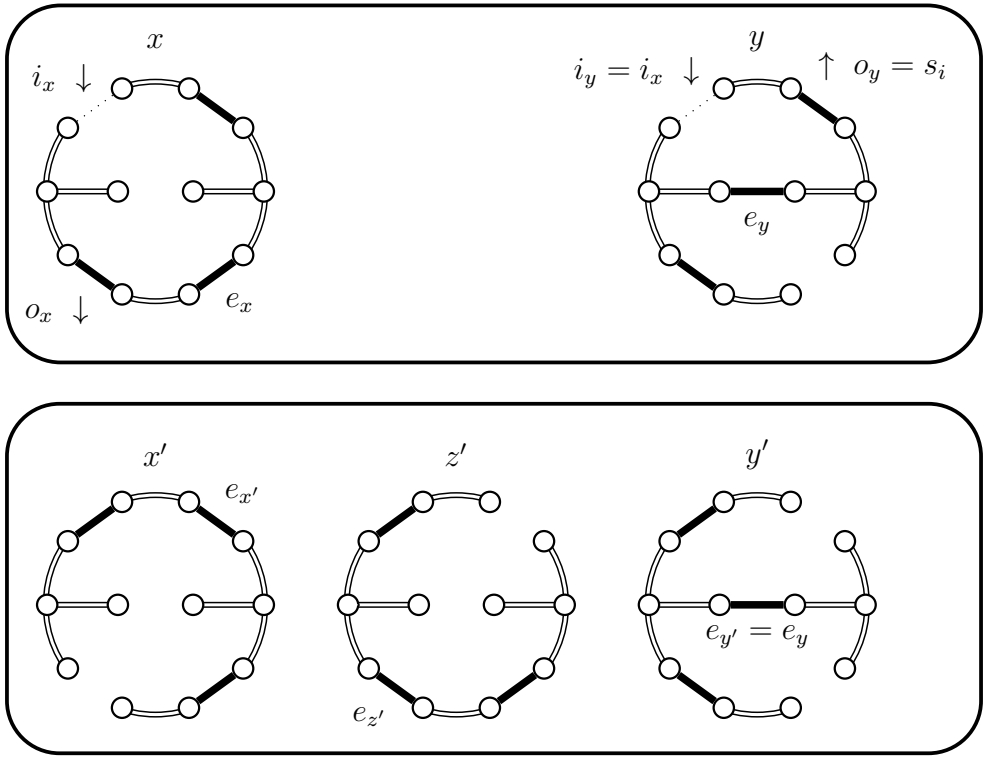}
          \renewcommand{\baseh}{
            \pscircle[style=generic]{1.5}

            \psarc[linecolor=white,linewidth=0.2](0,0){1.5}{36}{72}
            \Cnode[style=node](1.44;36){q1r}
            \Cnode[style=node](1.44;72){q1l}

            \psarc[linecolor=white,linewidth=0.2](0,0){1.5}{108}{144}
            \Cnode[style=node](1.44;108){q2r}
            \Cnode[style=node](1.44;144){q2l}

            \psarc[linecolor=white,linewidth=0.2](0,0){1.5}{216}{252}
            \Cnode[style=node](1.44;216){b1r}
            \Cnode[style=node](1.44;252){b1l}

            \psarc[linecolor=white,linewidth=0.2](0,0){1.5}{288}{324}
            \Cnode[style=node](1.44;288){b2r}
            \Cnode[style=node](1.44;324){b2l}

            \psline[style=generic](0.5,0)(1.5;0)
            \psline[style=generic](-0.5,0)(1.5;180)

            \Cnode[style=node](0.5,0){cr}
            \Cnode[style=node](-0.5,0){cl}

            \Cnode[style=node](1.44;0){ir}
            \Cnode[style=node](1.44;180){il}
          }

          \begin{center}
            \begin{pspicture}[showgrid=false](-2.0,-7.5)(11.0,2.5) \SpecialCoor
              \psframe[framearc=.3,linewidth=0.05](-2.0,-2.3)(11.0,2.5)
              \psframe[framearc=.3,linewidth=0.05](-2.0,-7.5)(11.0,-2.8)
              \rput(0,0){%
                \baseh \rput(0,2.0){$x$} \ncline[style=edge]{-}{q1r}{q1l}
                \ncline[style=edge]{-}{b2r}{b2l} \nbput{$e_x$}
                \ncline[style=insert]{-}{q2r}{q2l} \nbput{$i_x \ \downarrow$}
                \ncline[style=edge]{-}{b1r}{b1l} \nbput{$o_x \ \downarrow$} }
              \rput(8.0,0){ \baseh \rput(0,2.0){$y$}
                \ncline[style=edge]{-}{b1r}{b1l}
                \ncline[style=insert]{-}{q2r}{q2l}
                \nbput{$i_y = i_x \ \downarrow$} \ncline[style=edge]{-}{cl}{cr}
                \nbput{$e_y$} \ncline[style=edge]{-}{q1r}{q1l}
                \nbput{$\uparrow \ o_y = s_i $} }
              \rput(0,-5.5){%
                \baseh \rput(0,2.0){$x'$} \ncline[style=edge]{-}{q1r}{q1l}
                \nbput{$e_{x'}$} \ncline[style=edge]{-}{b2r}{b2l}
                \ncline[style=edge]{-}{q2r}{q2l} }
              \rput(4.0,-5.5){ \baseh \rput(0,2.0){$z'$}
                \ncline[style=edge]{-}{b2r}{b2l} \ncline[style=edge]{-}{q2r}{q2l}
                \ncline[style=edge]{-}{b1r}{b1l} \nbput{$e_{z'}$} }
              \rput(8.0,-5.5){%
                \baseh \rput(0,2.0){$y'$} \ncline[style=edge]{-}{b1r}{b1l}
                \ncline[style=edge]{-}{q2r}{q2l} \ncline[style=edge]{-}{cl}{cr}
                \nbput{$e_{y'}=e_y$} }
            \end{pspicture}
          \end{center}
          \caption{Case \ref{item:BigUn}, when $B^*$ is true.}
          \label{fig:BigUn}
        \end{figure}
      \item If $o_x \in E_x \setminus \{e_x\}$ and $B^*$ fails we use
        another Bernoulli random variable $B'$ with a success probability
        $p'$ defined as follows:
        \begin{equation*}
          p' = 1 - \frac{(C_x-1)(E_y-1)}{(C_y-1)(E_x-1)(1-p^*)}
        \end{equation*}
        In Lemma~\ref{lemma:Markovian} we prove that $B'$ properly balances
        the necessary probabilities. In case $B'$ yields \texttt{true}, use
        case~\ref{item:EqUy} and set $o_y = e_y$, see
        Figure~\ref{fig:EqUy}. Otherwise, if $s_y \in U_y$, use
        case~\ref{item:EqUy} (Figure~\ref{fig:EqUy}) or, if
        $s_y \notin U_y$, use case~\ref{item:EqUnC} (Figure~\ref{fig:EqUnC}).
        \label{item:BigBB}
      \end{enumerate}
    \end{enumerate}
  \end{enumerate}

  Notice that the case~\ref{item:EqCommon} applies when $C_x = C_y$, thus
  solving this situation as a particular case. This case is shown in
  Figure~\ref{fig:EqCommonNix}. It may even be the case that $C_x = C_y$ and
  $C_x$ and $C_e$ are disjoint, i.e., $C_x \cap C_e = \emptyset$. This case
  is not drawn.

  Formally, a coupling is Markovian when Equations~(\ref{eq:2})
  and~(\ref{eq:3}) hold, where $Z_t$ is the coupling, which is defined as a
  pair of chains $(X_t, Y_t)$. The chain $M_t$ represents the original
  chain.
\begin{subequations}
  \begin{align}
    \Pr (X_{t+1} = x'\ |\ Z_t=(x,y)) &= \Pr (M_{t+1} = x'\ |\
                                       M_t=x) \label{eq:2}
    \\
    \Pr (Y_{t+1} = y'\ |\ Z_t=(x,y)) &= \Pr (M_{t+1} = y'\ |\
                                       M_t=y) \label{eq:3}
  \end{align}
\end{subequations}
To establish vital insight into the coupling structure we will start by
studying it when it is Markovian.
\begin{lemma}
\label{lemma:Markovian}
  When $U_y = \{e_y\}$ the process we described is a Markovian coupling.
\end{lemma}
\begin{proof}
  The coupling verifies Equation~\eqref{eq:2}, because we do not alter the
  behavior of the chain $X_t$. Hence the main part of the proof focus on
  Equation~\eqref{eq:3}.

  First let us prove that for any edge $i \in E$ the probability that
  $i_y = i$ is $1/E$, i.e., $\Pr(i_y = i) = 1/E$. The possibilities for
  $i_y$ are the following:
  \begin{itemize}
  \item $i \in A_y$, this occurs only in case~\ref{item:loop}, when
    $i_x \in A_x$. It may be that $i = e_y$, this occurs when $i_x = e_x$
    in which case $i_y = e_y = i$ and this is the only case where $i_y =
    e_y$. In this case $\Pr(i_y = i) = \Pr(i_x = e_x) = 1/E$. Otherwise
    $i \in A_y \cap A_x$, in these cases $i_y = i_x$ and therefore
    $\Pr(i_y = i) = \Pr(i_x = i) = 1/E$.
  \item $i = e_x$, this occurs in
    cases~\ref{item:swap}~and~\ref{item:ExICol}, i.e., when $i_x = e_y$,
    which is the decisive condition for this choice. Therefore
    $\Pr(i_y = i) = \Pr(i_x = e_y) = 1/E$.
  \item $i \in E \setminus A_y$, this occurs in case~\ref{item:GenI}. In
    this case $i_y = i_x$ so again we have that
    $\Pr(i_y = i) = \Pr(i_x = i) = 1/E$.
  \end{itemize}
  Before focusing on $o_x$ we will prove that the Bernoulli random variables
  are well defined, i.e., that the expression on the denominators are not
  $0$ and that the values of $p$, $p^*$, $p'$ are between $0$ and $1$.
  \begin{itemize}
  \item Analysis of $B$. We need to have $C_y-1 \neq 0$ for $p$ to be well
    defined. Any cycle must contain at least $3$ edges, therefore
    $3 \leq C_y$ and hence $0 < 2 \leq C_y -1$. This guarantees that the
    denominator is not $0$. The same argument proves that $0 < C_x -1$,
    thus implying that $0 < p$, as both expressions are positive. We
    also establish that $p < 1$ because of the hypothesis of
    case~\ref{item:Smx} which guarantees $C_x < C_y$ and therefore
    $C_x - 1 < C_y - 1$.
  \item Analysis of $B^*$.  As in seen the analysis of $B$ we have that
    $0 < C_y-1$ and $0 < C_x-1$, therefore those denominators are not
    $0$. Moreover we also need to prove that $E_x -1 \neq 0$. In general we
    have that $1 \leq E_y$, because $e_y \in E_y$. Moreover, the hypothesis
    of case~\ref{item:BigUn} is that $C_y < C_x$ and therefore $E_y < E_x$,
    obtained by removing $I$ from the both sides. This implies that
    $1 < E_x$ and therefore $0 < E_x -1$, thus establishing that the last
    denominator is also not $0$.

    Let us now establish that $0 \leq p^*$ and $p^* < 1$. Note that $p^*$
    can be simplified to the expression
    \mbox{$(C_x-C_y)(I-1)/((C_y-1)(E_x-1))$}, where all the expressions in
    parenthesis are non-negative, so $0 \leq p^*$. For the second property
    we use the new expression for $p^*$ and simplify $p^* < 1$ to
    $(E_x - E_y)(I-1) < (E_x-1)(C_y-1)$. The deduction is straightforward
    using the equality $C_x - C_y = E_x - E_y$ that is obtained by removing
    $I$ from the left side. The properties $E_x - E_y \leq E_x - 1$ and
    $I - 1 < C_y - 1$ establish the desired result.
  \item Analysis of $B'$. We established, in the analysis of $B$, that
    $C_y-1$ is non-zero. In the analysis of $B^*$ we also established
    that $E_x-1$ is non-zero, note that case~\ref{item:BigBB} also assumes
    the hypothesis that $C_y < C_x$. Moreover in the analysis of $B^*$ we
    also established that $p^* < 1$, which implies that $0 < 1- p^*$ and
    therefore the last denominator is also non-zero.

    Let us also establish that $0 \leq p'$ and $p' \leq 1$. For the
    second property we instead prove that $0 \leq 1-p'$, where
    $1-p' = (C_x-1)(E_y-1)/((C_y-1)(E_x-1)(1-p^*))$ and all of the
    expressions in parenthesis are non-negative. We use the following
    deduction of equivalent inequalities to establish that $0 \leq p'$:
    \begin{align*}
      0 & \leq p' \\
      -p' & \leq 0 \\
      1-p' & \leq 1 \\
      (C_x-1)(E_y-1) &\leq (C_y-1)(E_x-1)(1-p^*) \\
      (C_x-1)(E_y-1) &\leq (C_y-1)(E_x-1)\left(1-
                       \frac{(C_x-C_y)(I-1)}{(C_y-1)(E_x-1)}\right) \\
      (C_x-1)(E_y-1) &\leq (C_y-1)(E_x-1)-(C_x-C_y)(I-1) \\
      (E_x-1)(E_y-1) + I(E_y-1) &\leq (E_y-1)(E_x-1)+ I(E_x-1) -
                                  (E_x-E_y)(I-1) \\
      I(E_y-1) &\leq  I(E_x-1) - (E_x-E_y)(I-1) \\
      I((E_y-1) + E_x - E_y) &\leq  I(E_x-1) + E_x-E_y \\
      I(E_x-1) + E_y &\leq  I(E_x-1) + E_x \\
      E_y &\leq E_x \\
      C_y &\leq C_x
    \end{align*}
    This last inequality is part of the hypothesis of
    case~\ref{item:BigBB}.

  \end{itemize}
  Now let us focus on the edge $o_x$. We wish to establish that for any
  $o \in C_y \setminus \{ i_y \}$ we have that $\Pr(o_y = o) =
  1/(C_y-1)$. We analyse this edge according to the following cases:
  \begin{enumerate}
  \item When the cycles are equal $C_x = C_y$. This involves
    cases~\ref{item:swap}~and~\ref{item:ExICol}.
    \begin{itemize}
    \item $o = e_y$, this occurs only in case~\ref{item:swap} and it is
      determined by the fact that $o_x = e_x$, therefore $\Pr(o_y = o) =
      \Pr(o_x = e_x) = 1 /(C_x - 1) = 1/(C_y -1)$.
    \item $o \neq e_y$, this occurs only in case~\ref{item:ExICol} and it
      is determined by the fact that $o_x \neq e_x$, in this case
      $o_y = o_x$. Therefore
      $\Pr(o_y = o) = \Pr(o_x = o) = 1 /(C_x - 1) = 1/(C_y -1)$.
    \end{itemize}
  \item When the cycles have the same size $|C_x| = |C_y|$,
    case~\ref{item:Eqs}. The possibilities for $o$ are the following:
    \begin{itemize}
    \item $o = e_y$, this occurs only in the case~\ref{item:EqOeX}. This
      case is determined by the fact that $o_x = e_x$. Therefore
      $\Pr(o_y = o) = \Pr(o_x = e_x) = 1 /(C_x - 1) = 1/(C_y -1)$.  Note
      that according to the Lemma's hypothesis, case~\ref{item:EqUy} never
      occurs.
    \item $o \in I \setminus \{i_y\}$, this occurs only in
      case~\ref{item:EqCommon}. This case is determined by the fact that
      $o_x \in I \setminus \{ i_x\}$ and sets $o_y = o_x = o$. Therefore
      $\Pr(o_y = o) = \Pr(o_x = o) = 1 /(C_x - 1) = 1/(C_y -1)$.
    \item $o \in E_y \setminus \{e_y\}$, this occurs only in
      case~\ref{item:EqUnC}. This case is determined by the fact that
      $o_x \in X \setminus \{e_x\}$ and moreover sets $o_y = s_y$, which
      was uniformly selected from $E_y \setminus \{e_y\}$. We have the
      following deduction where we use the fact that the events are
      independent and that $|C_x| = |C_y|$ implies $|E_x| = |E_y|$:
      \begin{align*}
        \Pr(o_y = o) & =  \Pr(o_x \in E_x \setminus \{e_x\} \mbox{ and }
                       s_y = o) \\
                     & = \Pr(o_x \in E_x \setminus \{e_x\}) \Pr(s_y = o) \\
                     & = \frac{E_x-1}{C_x - 1} \times \frac{1}{E_y-1}  \\
                     & = 1/(C_x - 1) \\
                     & = 1/(C_y - 1)
      \end{align*}

    \end{itemize}
  \item When $C_x < C_y$ this involves case~\ref{item:Smx}. The cases
    for $o$ are the following:
     \begin{itemize}
     \item $o = e_y$, this occurs only in the case~\ref{item:SmUnC} and
       when $B$ is \texttt{true}. This case occurrs when $o_x = e_x$.
        We make the following deduction, that uses the fact that the events
        are independent and the success probability of $B$:
       \begin{align*}
        \Pr(o_y = o) & =  \Pr(o_x = e_x \mbox{ and }
                       B = \mathtt{true}) \\
                     & = \Pr(o_x = e_x) \Pr(B = \mathtt{true}) \\
                     & = \frac{1}{C_x - 1} \times \frac{C_x-1}{C_y-1}  \\
                     & = 1/(C_y - 1)
      \end{align*}
    \item $o \in I \setminus \{i_y\}$, this occurs only in
      case~\ref{item:SmCommonUy} and when $B$ is \texttt{true}. This case
      is determined by the fact that $o_x \in I \setminus \{ i_x\}$ and
      sets $o_y = o_x = o$.  We make the following deduction, that uses the
      fact that the events are independent and the success probability of
      $B$:
       \begin{align*}
        \Pr(o_y = o) & =  \Pr(o_x = o \mbox{ and }
                       B = \mathtt{true}) \\
                     & = \Pr(o_x = o) \Pr(B = \mathtt{true}) \\
                     & = \frac{1}{C_x - 1} \times \frac{C_x-1}{C_y-1}  \\
                     & = 1/(C_y - 1)
      \end{align*}
    \item $o \in E_y \setminus \{e_y\}$, this occurs in
      case~\ref{item:SmUnR}, but also in cases~\ref{item:SmComUncommon}
      and~\ref{item:SmUnC} when $B$ is \texttt{false}. We have the
      following deduction, that uses event independence, the fact that the
      cases are disjoint events and the succes probability of $B$:
      \begin{align*}
        \Pr&(o_y = o) & \\
&        = \Pr( \mbox{\ref{item:SmUnR}
        or~(\ref{item:SmComUncommon} and $B = \mathtt{false}$) or~(\ref{item:SmUnC} and $B = \mathtt{false}$)}) \\
&        = \Pr( \mbox{\ref{item:SmUnR}})
        + \Pr(\mbox{\ref{item:SmComUncommon}} \mbox{ and } B = \mathtt{false}) +
        \Pr(\mbox{\ref{item:SmUnC}} \mbox{ and } B = \mathtt{false}) \\
&      = \Pr(o_x \in E_x \setminus \{e_x\} \mbox{ and }
       s_y = o) + \Pr(o_x \in I \mbox{ and } B =
       \mathtt{false} \mbox{ and } s_y = o) \\
& \quad      + \Pr(o_x = e_x \mbox{ and } B =
       \mathtt{false} \mbox{ and } s_y = o) \\
&      = \Pr(o_x \in E_x \setminus \{e_x\}) \Pr(s_y = o) +
       \Pr(o_x \in I) \Pr(B = \mathtt{false})\Pr(s_y = o) \\
& \quad     + \Pr(o_x = e_x)\Pr(B = \mathtt{false})\Pr(s_y = o) \\
&      = \Pr(o_x \in E_x \setminus \{e_x\}) \Pr(s_y = o) +
       \Pr(o_x \in I \cup \{e_x\}) \Pr(B = \mathtt{false})\Pr(s_y = o) \\
&      = [ \Pr(o_x \in E_x \setminus \{e_x\}) +
       \Pr(o_x \in I \cup \{e_x\}) (1 - \Pr(B =
       \mathtt{true})) ]\Pr(s_y = o) \\
&      = [\Pr(o_x \in C_x) -
       \Pr(o_x \in I \cup \{e_x\}) \Pr(B =
       \mathtt{true}) ] \Pr(s_y = o) \\
&      = [ 1 -
       \Pr(o_x \in I \cup \{e_x\}) \Pr(B =
       \mathtt{true}) ] \Pr(s_y = o) \\
&      = \left [ 1 -\frac{I-1+1}{C_x-1} \times \frac{C_x-1}{C_y-1} \right]
       \Pr(s_y = o) \\
&      = \left[ 1 -\frac{I}{C_y-1} \right] \Pr(s_y = o) \\
&      = \frac{C_y - 1 -I}{C_y-1} \times \frac{1}{E_y-1} \\
&      = \frac{1}{C_y-1}
      \end{align*}
    \end{itemize}
  \item When $C_x > C_y$, this concerns case~\ref{item:Bigx}. The cases
    for $o$ are the following:
    \begin{itemize}
    \item $o = e_y$, this occurs in the case~\ref{item:BigEx} and
      case~\ref{item:BigBB} when $B'$ is \texttt{true}. We use the
      following deduction:
      \begin{align*}
        \Pr&(o_y = o) \\
           & = \Pr(\mbox{\ref{item:BigEx} or~(\ref{item:BigBB} and $B' = \mathtt{true}$)})\\
           & = \Pr(\mbox{\ref{item:BigEx}}) + \Pr(\mbox{\ref{item:BigBB} and } B' = \mathtt{true})\\
           & = \Pr(o_x = e_x) + \Pr(o_x \in E_x \setminus \{e_x\} \mbox{ and
             } B^* = \mathtt{false} \mbox{ and } B' = \mathtt{true})\\
           & = \frac{1}{C_x-1} + \frac{E_x-1}{C_x-1}
             (1-p^*)\left(1-\frac{(C_x-1)(E_y-1)}{(C_y-1)(E_x-1)(1-p^*)}\right)
        \\
           & = \frac{1}{C_x-1} + \frac{E_x-1}{C_x-1}
             \left(1-p^*-\frac{(C_x-1)(E_y-1)}{(C_y-1)(E_x-1)}\right) \\
           & = \frac{1}{C_x-1} + 
             \frac{E_x-1}{C_x-1}-\frac{E_x-1}{C_x-1}p^*-\frac{E_y-1}{C_y-1} \\
           & = \frac{1}{C_x-1} + 
             \frac{E_x-1}{C_x-1}-\left[\frac{1}{C_y-1}-\frac{1}{C_x-1}\right](I-1)-\frac{E_y-1}{C_y-1}
        \\
           & = \frac{1}{C_x-1} + 
             \frac{E_x-1}{C_x-1}-\frac{I-1}{C_y-1}+\frac{I-1}{C_x-1}-\frac{E_y-1}{C_y-1}
             \\
           & = \frac{E_x + I - 1}{C_x-1} - 
             \frac{E_y + I -1}{C_y-1} + \frac{1}{C_y-1} \\
           & = \frac{C_x - 1}{C_x-1} - 
             \frac{C_y -1}{C_y-1} + \frac{1}{C_y-1} \\
           & = \frac{1}{C_y-1} \\
      \end{align*}
    \item $o \in I \setminus \{i_y\}$, this occurs in
      case~\ref{item:BigCommon} and case~\ref{item:BigUn} when $B^*$ is
      \texttt{true}. We make the following deduction
      \begin{align*}
        \Pr&(o_y = o) \\
           & = \Pr(\mbox{\ref{item:BigCommon} or~\ref{item:BigUn}})\\
           & = \Pr(\mbox{\ref{item:BigCommon}}) + \Pr(\mbox{\ref{item:BigUn}})\\
           & = \Pr(o_x = o) + \Pr(o_x \in E_x \setminus \{e_x\} \mbox{ and
             } B^* = \mathtt{true} \mbox{ and
             } s_i = o)\\
           & = \Pr(o_x = o) + \Pr(o_x \in E_x \setminus \{e_x\})
             \Pr(B^* = \mathtt{true}) \Pr(s_i = o) \\
           & = \frac{1}{C_x-1} + \frac{E_x - 1}{C_x-1} \times
                       \left( \frac{1}{C_y-1} - \frac{1}{C_x-1}\right)
             \times \frac{(C_x-1)(I-1)}{E_x-1}
             \times \frac{1}{I-1}\\
           & = \frac{1}{C_x-1} + \frac{1}{C_y-1} - \frac{1}{C_x-1} \\
           & = \frac{1}{C_y-1} \\
      \end{align*}
    \item $o \in E_y \setminus \{e_y\}$, this occurs in
      case~\ref{item:BigBB} when $B'$ is \texttt{false}. We have the
      following deduction:
      \begin{align*}
        \Pr&(o_y = o) \\
           & = \Pr(\ref{item:BigBB} \mbox{ and } B' = \mathtt{false})\\
           & = \Pr(o_x \in E_x \setminus \{e_x\} \mbox{ and
             } B^* = \mathtt{false} \mbox{ and } B' = \mathtt{false}
             \mbox{ and } s_y = o)\\
           & = \Pr(o_x \in E_x \setminus \{e_x\})\Pr(B^* =
             \mathtt{false})\Pr(B' = \mathtt{false})\Pr(s_y = o)\\
           & = \frac{E_x-1}{C_x-1}
             (1-p^*)\frac{(C_x-1)(E_y-1)}{(C_y-1)(E_x-1)(1-p^*)} \times
             \frac{1}{E_y-1}\\
           & = \frac{1}{C_y-1}
      \end{align*}
    \end{itemize}
  \end{enumerate}
\end{proof}
\begin{lemma}
\label{lemma:nonMarkovian}
 The process we described is a non-Markovian coupling.
\end{lemma}
\begin{proof}
  In the context of a Markovian coupling we analyse the transition from $y$
  to $y'$ given the information about $x$. In the non-Markovian case we
  will use less information about $x$. We assume that $e_y$ is a random
  variable and that $x$ provides only $e_x$ and $U_y$ and we know only that
  $e_y \in U_y$ and moreover that $\Pr(e_y = e) = 1/U_y$, for any
  $e \in U_y$ and $\Pr(e_y = e) = 0$ otherwise. Then the chain $X_t$ makes
  its move and provides information about $i_x$ and $o_x$. Let us consider
  only the cases when $Y_t$ then chooses $i_y = i_x$, because nothing
  changes in the cases where this does not happen. Now $i_x$ can be used to
  define $C_x$ and $C_y$ and, therefore, $E_x$ and $E_y$. We focus our
  attention on $E_y \cap U_y$ because, except for the trivial cases, we
  must have $e_y \in (E_y \cap U_y)$. Hence, we instead alter our condition
  to $\Pr(e_y = e) = 1/|E_y \cap U_y|$, for any $e \in (E_y \cap U_y)$ and
  $0$ otherwise.

  Note that this is a reasonable process because we established in
  Lemma~\ref{lemma:CePartition} that, in the non-trivial case, $E_x$ and
  $E_y$ partition $C_e$ and, therefore, because
  $U_y \subseteq C_e \setminus \{e_x\}$, we have that $E_x$ and $E_y$ also
  partition $U_y$. This means that $U_y \cap I = \emptyset$ and so
  we are not loosing any part of $U_y$ in this process, we are only
  dividing it into cases. This process is also the reason why, even when
  $s_y \notin U_y$ we define $U_{y'} = (U_y \cap C_y) \setminus \{o_y\} =
  (U_y \cap E_y) \setminus \{o_y\}$.

  Now the cases considered in Lemma~\ref{lemma:Markovian} must be
  changed. Substitute the original cases of $o = e_y$ for
  $o \in U_y \cap E_y$. Also, substitute the case
  $o \in E_y \setminus \{ e_y\}$ by $o \in E_y \setminus U_y$. The other
  cases remain unaltered. Except for the first case, the previous
  deductions still apply. We will exemplify how the deduction changes
  for $o \in U_y \cap E_y$. We consider only the situation when
  $|C_x| = |C_y|$. For the remaining situations, $C_x < C_y$ and
  $C_x > C_y$, we use a general argument. Hence our precise assumptions
  are: $o \in U_y \cap E_y$ and $|C_x| = |C_y|$.
  \begin{align*}
    \frac{\Pr(o_y = o)}{\Pr(e_y = o)} & = \Pr(o_x \in E_x \setminus \{e_x\} \mbox{ and } s_y
                                        \in U_y) + \Pr(o_x = e_x)\\
                                      & = \Pr(o_x \in E_x \setminus \{e_x\}) \Pr(s_y
                                        \in U_y) + \Pr(o_x = e_x) \\
                                      & = \frac{E_x-1}{C_x-1} \times
                                        \frac{|U_y \cap E_y|-1}{E_y-1} +
                                        \frac{1}{C_x-1} \\
                                      & = \frac{|U_y \cap E_y|}{C_x-1} \\
                                      & = \frac{|U_y \cap E_y|}{C_y-1}
  \end{align*}
  Which is correct according to our assumption.

  The general argument follows the above derivation. Whenever the Markovian
  coupling would produce $o_y = e_y$, we obtain $1/(C_y-1)$ probability.
  Moreover, for every edge such that $o_y \in E_y \setminus \{e_y\}$
  produced by the Markovian coupling we obtain another $1/(C_y-1)$
  probability. This totals to $|U_y \cap E_y|/(C_y-1)$ as desired. This
  also occurs in the cases when $C_x < C_y$ and $C_x > C_y$, only the
  derivations become more cumbersome.

  Finally will argue why the property that $\Pr(e_y = e) = 1/U_y$ when
  $e \in U_y$. The set $U_y$ is initialised to contain only the edge $e_y$,
  i.e., $U_y = \{ e_y \}$. As the coupling proceeds $U_{y'}$ is chosen to
  represent the edges, from which $e_{y'}$ was chosen, or is simply
  restricted if $e_y$ does not change. More precisely: in
  case~\ref{item:EqUy} we chose $e_{y'} = o_x$; in case~\ref{item:SmUnR} we
  chose $e_{z'} = o_x$; in case~\ref{item:SmCommonUy} we chose $e_{y'} =
  o_x$.
\end{proof}

We obtained no general bounds on the coupling we presented, it may even be
that such bounds are exponential even if the Markov chain has polynomial
mixing time. In fact~\citet*{814596} proved that this is the case for
Markovian couplings of the Jerrum-Sinclair
chain~\citep{DBLP:journals/siamcomp/JerrumS89}. Note that according to the
classification of~\citet*{814596} the coupling we present is considered as
time-variant Markovian. Hence their result applies to type of coupling we
are using, albeit we are considering different chains so it is not
immediate that indeed there exist no polynomial Markovian couplings for the
chain we presented.  Cycle graphs are the only class of graphs for which we
establish polynomial bounds, see Figure~\ref{fig:cycleex}.

\begin{theorem}
\label{teo:cycleC}
  For any cycle graph $G$ the mixing time $\tau$ of edge swap chain is
  $O(V)$ for the normal version of the chain and $1/\log_4(V-1)$ for the
  fast version.
\end{theorem}
\begin{proof}
  For any two trees $A_x$ and $A_y$ we have a maximum distance of $1$ edge,
  i.e., $d(A_x, A_y) \leq 1$. Hence our coupling applies directly.

  For the fast version, case~\ref{item:loop} does not occur, because $i_x$
  is chosen from $E \setminus \{A_x\}$. Hence, the only cases that might
  apply are cases~\ref{item:swap} and~\ref{item:ExICol}. In first case, the
  chains preserve their distance and in the last case the distance is
  reduced to $0$. Hence, $E[d(X_1,Y_1)] = 1/(V-1)$, which corresponds to the
  probability of case~\ref{item:swap}. Each step of the coupling is
  independent, which means we can use the previous result and Markov's
  inequality to obtain
  $\Pr (d(X_\tau, Y_\tau) \geq 1) \leq 1/(V-1)^{\tau}$. Then, we use this
  probability in the coupling Lemma~\ref{lem:coupling} to obtain a
  variation distance of $1/4$, by solving the following equation:
  $1/(V-1)^{\tau} = 1/4$.

  For the slow version of the chain, case~\ref{item:loop} applies most of
  the time, i.e., for $V-1$ out of $V$ choices of $i_x$. It takes
  $V-1$ steps for the standard chain to behave as the fast chain and,
  therefore, the time should be $(V-1)/\log_4(V-1)$.
\end{proof}
This result is in stark contrast with the alternative algorithms, random
walk and Wilson's (see Section~\ref{sec:related-work}), which require
$O(V^2)$ time~\citep{levin2017markov}. More recent algorithms are also at
least $O(V^{4/3})$ for this case.

Moreover when a graph is a connect set of cycles connected by bridges or
articulation points we
can also establish a similar result. Figure~\ref{fig:genC} shows one such
graph.

\begin{figure}[tb]
\begin{center}
  \begin{pspicture}[showgrid=false](8,10)
\psset{shadow=true}
\rput(3,5){
\Cnode(1;0){A}
\Cnode(1;120){B}
\Cnode(1;240){C}

\Cnode(2.5;240){D}
\rput(2.5;240){
  \Cnode(1.5;180){E}
  \Cnode(1.5;-90){F}
  \rput(1.5;-90){\rput(3,0){
\Cnode(1;0){O}
\Cnode(1;120){P}
\Cnode(1;240){Q}
}}
}

\rput(3.5;0){
  \Cnode(1;0){G}
  \Cnode(1;90){H}
  \rput(2;90){
    \Cnode(1;0){HA}
    \Cnode(1;90){HB}
    \Cnode(1;180){HC}
}
  \Cnode(1;180){I}
  \Cnode(1;-90){J}
}

\rput(3.5;120){
  \Cnode(1;120){K}
  \Cnode(1;210){L}
  \Cnode(1;300){M}
  \Cnode(1;30){N}
}
}

\psset{shadow=false}
\ncline{A}{B}
\ncline{A}{C}
\ncline{B}{C}

\ncline{E}{D}
\ncline{F}{D}
\ncline{C}{D}

\ncline{G}{H}
\ncline{H}{I}
\ncline{I}{J}
\ncline{J}{G}

\ncline{H}{HA}
\ncline{HA}{HB}
\ncline{HB}{HC}
\ncline{HC}{H}

\ncline{A}{I}

\ncline{K}{L}
\ncline{L}{M}
\ncline{M}{N}
\ncline{N}{K}

\ncline{B}{M}

\ncline{O}{P}
\ncline{P}{Q}
\ncline{Q}{O}

\ncline{F}{P}
\end{pspicture}
\end{center}
\caption{Cycles connected by bridges or articulation points.}
  \label{fig:genC}
\end{figure}
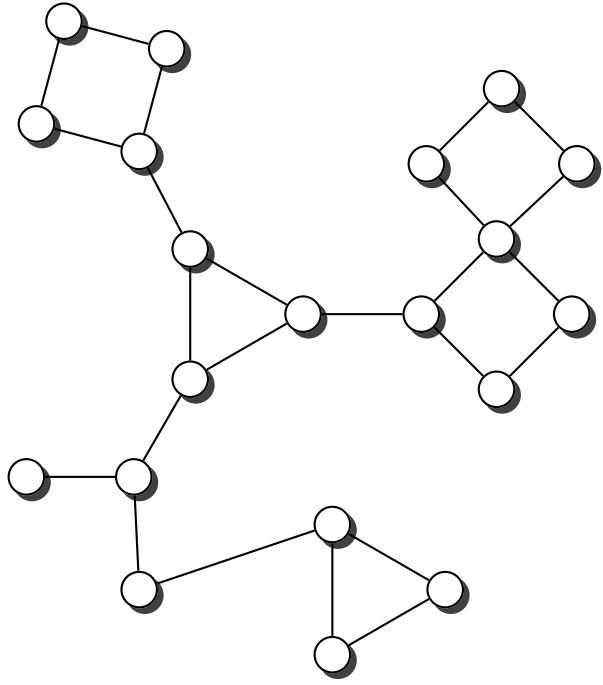
\begin{theorem}
\label{teo:cycleCpp}
For any graph $G$ which consists of $n$ simple cycles connect by bridges or
articulation points, such that $m$ is the size of the smallest cycle, then
the mixing time $\tau$ of the fast edge swap chain is the following:
\[
  \tau = \frac{\log (4n)}{\log \left(\frac{n(m-1)}{n(m-1)-(m-2)} \right)}
\]

The mixing time for the slow version is obtained by using $|E|$ instead of
$n$ in the previous expression.
\end{theorem}
\begin{proof}
  To obtain this result we use a path coupling argument. Then for two
  chains at distance $1$ we have
  $E[d(X_1,Y_1)] \leq \left(1 - \frac{m-2}{n(m-1)} \right)$.

We assume that the different edge occurs in the largest cycle. In general
the edges inserted and deleted do not alter this situation, hence the term
$1$. However with probability $1/n$ the chain $X_t$ inserts an edge that
creates the cycle where the diference occurs. In that case with probability
$(m-2)/(m-1)$ the chains coalesce. Hence applying path coupling the mixing
time must verify the following equation:
\[
n  \left(1 - \frac{m-2}{n(m-1)} \right)^\tau \leq \frac{1}{4}
\]

For the slow version the chain choose the correct edge with probability
$1/E$ instead of $1/n$.
\end{proof}

\subsection{Experimental Results}
\label{sec:experimental-results}
\subsubsection{Convergence Testing}
\label{sec:convergence-testing}
Before looking at the performance of the algorithm we started by testing
the convergence of the edge swap chain. We estimate the variation
distance after a varying number of iterations. The results are shown in
Figures~\ref{fig:sparse},~\ref{fig:cycle},~\ref{fig:dense},~\ref{fig:biK},~\ref{fig:dmP05},~\ref{fig:dmP25}, and~\ref{fig:dmP50}. We
now describe the structure of these figures. Consider for example
Figure~\ref{fig:dense}. The structure is the following:
\begin{itemize}
\item The bottom left plot shows the graph properties, the number of
  vertexes $V$ in the $x$ axis and the number edges $E$ on the $y$
  axis. For the dense case graph $0$ has $10$ vertexes and $45$
  edges. Moreover, graph $6$ has $40$ vertexes and $780$ edges. These graph
  indexes are used in the remaining plots.
\item The top left plot show the number of iterations $t$ of the chain in
  the $x$ axis and the estimated variation distance on the $y$ axis, for
  all the different graphs.
\item The top right plot is similar to the top left, but the $x$ axis
  contains the number of iterations divided by $(V^{1.3}+E)$. Besides the
  data this plot also show a plot of $\ln(1/\hat{\varepsilon})$ for reference.
\item The bottom right plot is the same as the top right plot, using
a logarithmic scale on the $y$ axis.
\end{itemize}
To avoid the plots from becoming excessively dense, we do not plot points
for all experimental values, instead plot one point out of 3. However, the
lines pass through all experimental points, even those that are not explicit.

The variation distance between two distributions $D_1$ and $D_2$ on a
countable state space $S$ is given by
$\Vert D_1 - D_2 \Vert = \sum_{x \in S} |D_1(x) - D_2(x)|/2$. This is the
real value of $\varepsilon$. However, the size of $S$ quickly becomes larger
than we can compute. Instead, we compute a simpler variation distance
$\Vert D_1 - D_2 \Vert_d$, where $S$ is reduced from the set of all
spanning trees of $G$ to the set of integers from $0$ to $V-1$, which
correspond to the edge distance, defined in Section~\ref{sec:coping}, of
the generated tree $A$ to a fixed random spanning tree $R$. More precisely,
we generate $20$ random trees, using a random walk algorithm described in
Section~\ref{sec:related-work}. For each of these trees, we compute
$\Vert \pi - M_t \Vert_d$, i.e., the simpler distance between the
stationary distribution $\pi$ and the distribution $M_t$ obtained by
computing $t$ steps of the edge swapping chain. To obtain $M_t$, we start
from a fixed initial tree $A_0$ and execute our chain $t$ times. This
process is repeated several times to obtain estimates for the corresponding
probabilities. We keep two sets of estimates $M_t$ and $M'_t$ and stop when
$\Vert M_t - M'_t \Vert_d < 0.05$. Moreover we only estimate values where
$\Vert \pi - M_t \Vert_d \geq 0.1$. We use the same criteria to estimate
$\pi$, but in this case the trees are again generated by the random walk
algorithm. The final value $\hat{\varepsilon}$ is obtained as the maximum
value obtained for the $20$ trees.

We generated dense graphs, sparse graphs and some in between graphs. The
\textbf{sparse} graphs are ladder graphs; an illustration of these graphs is
shown in Figure~\ref{fig:ladder}. The \textbf{cycle} graphs consist of a
single cycle, as shown in Figure~\ref{fig:cycleex}.
\begin{figure}[tbp]
  \begin{center}
  \begin{pspicture}(-0.6,-0.6)(13.0,3.0)
     \psset{shadow=true}
     \psset{radius=4mm}
     \psset{linecolor=black}
     \psset{linewidth=0.3mm}
     \psset{linestyle=solid}
     \Cnode(0,0){a0}
     \Cnode(2,0){a1}
     \Cnode(4,0){a2}
     \Cnode(6,0){a3}
     \Cnode(8,0){a4}
     \Cnode(10,0){a5}
     \Cnode(12,0){a6}
     \Cnode(0,2){a7}
     \Cnode(2,2){a8}
     \Cnode(4,2){a9}
     \Cnode(6,2){a10}
     \Cnode(8,2){a11}
     \Cnode(10,2){a12}
     \Cnode(12,2){a13}
     \psset{linewidth=0.1}
     \psset{doubleline=false}
     \psset{shadow=false}

     \ncline{-}{a0}{a7}
     \ncline{-}{a1}{a8}
     \ncline{-}{a2}{a9}
     \ncline{-}{a3}{a10}
     \ncline{-}{a4}{a11}
     \ncline{-}{a5}{a12}
     \ncline{-}{a6}{a13}

     \ncline{-}{a0}{a1}
     \ncline{-}{a1}{a2}
     \ncline{-}{a2}{a3}
     \ncline{-}{a3}{a4}
     \ncline{-}{a4}{a5}
     \ncline{-}{a5}{a6}

     \ncline{-}{a7}{a8}
     \ncline{-}{a8}{a9}
     \ncline{-}{a9}{a10}
     \ncline{-}{a10}{a11}
     \ncline{-}{a11}{a12}
     \ncline{-}{a12}{a13}

  \end{pspicture}
  \end{center}
  \caption{A ladder graph. }
\label{fig:ladder}
\end{figure}
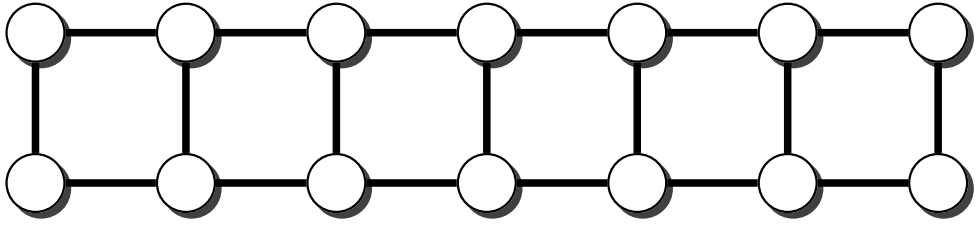
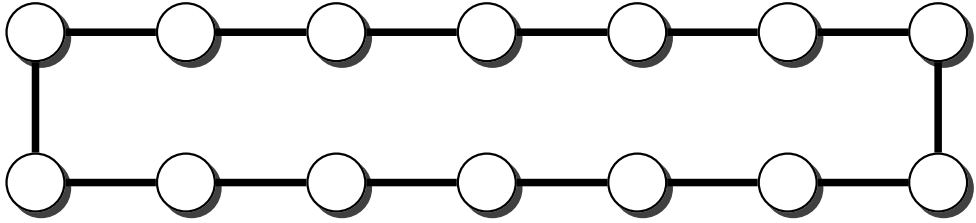
\begin{figure}[tbp]
  \begin{center}
  \begin{pspicture}(-0.6,-0.6)(13.0,3.0)
     \psset{shadow=true}
     \psset{radius=4mm}
     \psset{linecolor=black}
     \psset{linewidth=0.3mm}
     \psset{linestyle=solid}
     \Cnode(0,0){a0}
     \Cnode(2,0){a1}
     \Cnode(4,0){a2}
     \Cnode(6,0){a3}
     \Cnode(8,0){a4}
     \Cnode(10,0){a5}
     \Cnode(12,0){a6}
     \Cnode(0,2){a7}
     \Cnode(2,2){a8}
     \Cnode(4,2){a9}
     \Cnode(6,2){a10}
     \Cnode(8,2){a11}
     \Cnode(10,2){a12}
     \Cnode(12,2){a13}
     \psset{linewidth=0.1}
     \psset{doubleline=false}
     \psset{shadow=false}

     \ncline{-}{a0}{a7}
     \ncline{-}{a6}{a13}

     \ncline{-}{a0}{a1}
     \ncline{-}{a1}{a2}
     \ncline{-}{a2}{a3}
     \ncline{-}{a3}{a4}
     \ncline{-}{a4}{a5}
     \ncline{-}{a5}{a6}

     \ncline{-}{a7}{a8}
     \ncline{-}{a8}{a9}
     \ncline{-}{a9}{a10}
     \ncline{-}{a10}{a11}
     \ncline{-}{a11}{a12}
     \ncline{-}{a12}{a13}

  \end{pspicture}
  \end{center}
  \caption{A cycle graph.}
\label{fig:cycleex}
\end{figure}
The \textbf{dense} graphs
are actually the complete graphs $K_{V}$. We also generated other dense
graphs labelled \textbf{biK} which consisted of two complete graphs
connected by two edges. We also generated graphs based on the duplication
model \textbf{dmP}. Let $G_0 = (V_0, E_0)$ be an undirected and unweighted
graph. Given $0\leq p \leq 1$, the partial duplication model
builds a graph $G = (V, E)$ by partial duplication as follows~\citep{DBLP:journals/jcb/ChungLDG03}: start with
$G = G_0$ at time $t = 1$ and, at time $t > 1$, perform a duplication step:
\begin{enumerate}
\item Uniformly select a random vertex $u$ of $G$.
\item Add a new vertex $v$ and an edge $(u, v)$.
\item For each neighbor $w$ of $u$, add an edge $(v, w)$ with probability
  $p$.
\end{enumerate}
The different values at the end of {\bf dmP}, namely in Figures~\ref{fig:dmP05},~\ref{fig:dmP25} and~\ref{fig:dmP50}, correspond to the choices of $p$.

\begin{figure}[tbp]
  \begin{center}
    \hspace{-2cm}
    \scalebox{0.8}{        \input{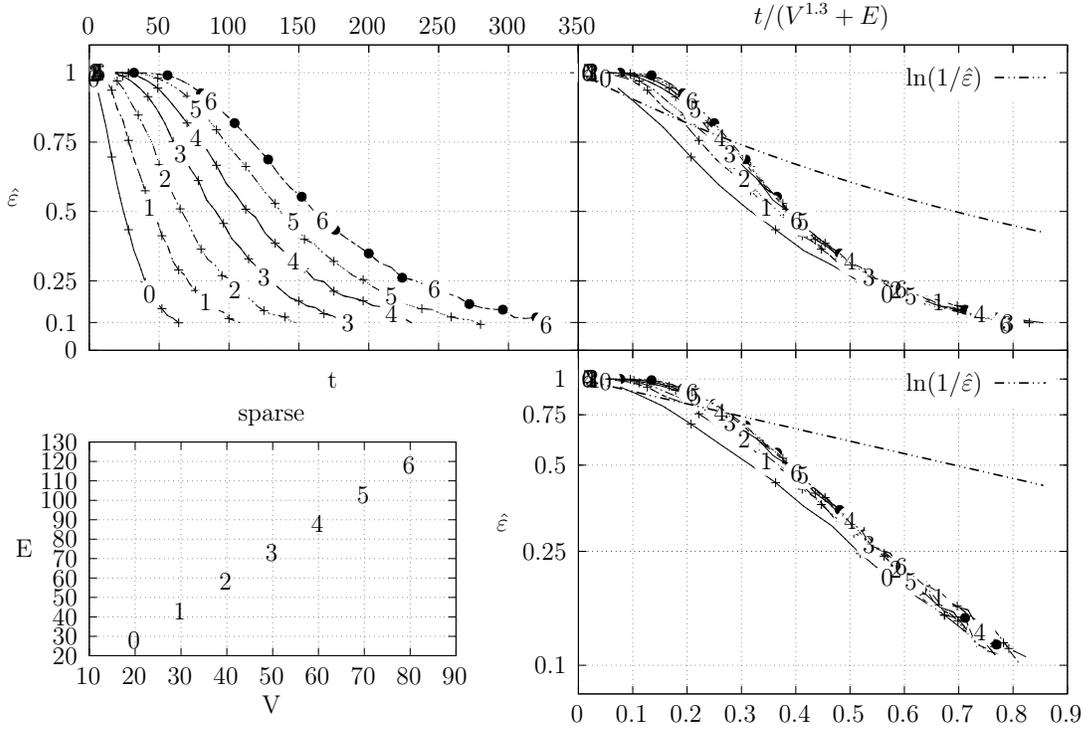}}
  \end{center}
  \caption{Estimation of variation distance as a function of the number of iterations for {\bf sparse} graphs (see Section~\ref{sec:convergence-testing} for details).}
  \label{fig:sparse}
\end{figure}
\begin{figure}[tbp]
  \begin{center}
    \hspace{-2cm}
    \scalebox{0.8}{        \input{cycle.tex}}
  \end{center}
  \caption{Estimation of variation distance as a function of the number of iterations for {\bf cycle} graphs (see Section~\ref{sec:convergence-testing} for details).}
\label{fig:cycle}
\end{figure}
\begin{figure}[tbp]
  \begin{center}
    \hspace{-2cm}
    \scalebox{0.8}{        \input{dense.tex}}
  \end{center}
  \caption{Estimation of variation distance as a function of the number of iterations for {\bf dense} graphs (see Section~\ref{sec:convergence-testing} for details).}
\label{fig:dense}
\end{figure}
\begin{figure}[tbp]
  \begin{center}
    \hspace{-2cm}
    \scalebox{0.8}{        \input{biK.tex}}
  \end{center}
  \caption{Estimation of variation distance as a function of the number of iterations for {\bf biK} graphs (see Section~\ref{sec:convergence-testing} for details).}
\label{fig:biK}
\end{figure}
\begin{figure}[tbp]
  \begin{center}
    \hspace{-2cm}
    \scalebox{0.8}{    \input{dmP05.tex}}
  \end{center}
  \caption{Estimation of variation distance as a function of the number of iterations for {\bf dmP} graphs (see Section~\ref{sec:convergence-testing} for details).}
  \label{fig:dmP05}
\end{figure}
\begin{figure}[tbp]
  \begin{center}
    \hspace{-2cm}
    \scalebox{0.8}{    \input{dmP25.tex}}
  \end{center}
  \caption{Estimation of variation distance as a function of the number of iterations for {\bf dmP} graphs (see Section~\ref{sec:convergence-testing} for details).}
  \label{fig:dmP25}
\end{figure}
\begin{figure}[tbp]
  \begin{center}
    \hspace{-2cm}
    \scalebox{0.8}{    \input{dmP50.tex}}
  \end{center}
  \caption{Estimation of variation distance as a function of the number of iterations for {\bf dmP} graphs (see Section~\ref{sec:convergence-testing} for details).}
  \label{fig:dmP50}
\end{figure}
These graphs show the convergence of the Markov chain and moreover
$V^{1.3}+E$ seems to be a reasonable bound for $\tau$. Still, these results
are not entirely binding. On the one hand the estimation of the variation
distance groups several spanning trees into the same distance, which means
that within a group the distribution might not be uniform, even if the
global statistics are good. So the actual variation distance may be larger
and the convergence might be slower. On the other hand we chose the
exponent 1.3 experimentally by trying to force the data of the graphs to
converge at the same point. The actual value may be smaller or larger.
\subsubsection{Coupling Simulation}
\label{sec:coupling-simulation}
As mentioned before, we obtained no general bounds on the coupling we presented. In
fact, experimental simulation for the coupling does not converge for all
classes of graphs. We obtained experimental convergence for cycle graph, as
expected from Theorem~\ref{teo:cycleC}, and for ladder graphs. For the
remaining graphs we used an optimistic version of the coupling which always
assumes that $s_y \in U_y$ and that $B^*$ fails. With these assumptions, all
the cases which increase the distance between states are eliminated and the
coupling always converges. Note that this approach does not yield a sound
coupling, but in practice we verified that this procedure obtained good
experimental variation distance. Moreover, the variation distance estimation
for these tests is not the simpler distance but the actual experimental
variation distance, obtained by generating several experimental trees, such
that in average each possible tree is obtained 100 times.

The simulation of the path coupling proceeds by generating a path with
$e \ln V$ steps, essentially selecting two trees at distance
$e \ln V$ from each other. This path is obtained by computing $e \ln V$
steps of the fast chain. Recall that our implementation and all simulations
use the fast version of the chain. The simulation ends once this path
contracts to size $\ln V$. Let $t'$ be the number of steps in this
process. Once this point is obtained our estimate for mixing time is
$\hat{\tau} = t' \ln V$. In general, we wish to obtain $\hat{\tau}$ such
that the probability that the two general chains $X_t$ and $X_t$ coalesce
is at least $75\%$. Hence, we repeat this process $4$ times and choose the
second largest value of $\hat{\tau}$ as our estimate.

Table~\ref{table:vd} summarizes results for the
experimental variation distance. The number of possible spanning trees
for each graph was computed through the Kirchoff's theorem. Then, we
generated 100 times the number of possible trees, and we computed the
variation distance. As stated above, we got good results for the
variation distance, getting a median well below 25\% for all tested graph topologies.
\begin{table}[tbp]
\caption{Variation distance (VD) for different graph topologies.
Median and maximum VD computed over 5 runs for each network. Since
{\bf dmP} graphs are random, results for {\bf dmP} were further
computed over 5 different graphs for each size $|V|$.}
\label{table:vd}
\begin{center}
\begin{tabular}{|c|c|c|c|}
\hline
Graph        & $|V|$          & median VD & max VD \\
\hline
{\bf dense}  & $\{5,7\}$      & 0.060     & 0.194  \\
{\bf biK}    & $\{8,10\}$     & 0.065     & 0.190  \\
{\bf cycle}  & $\{16,20,24\}$ & 0.001     & 0.004  \\
{\bf sparse} & $\{10,14,20\}$ & 0.053     & 0.110  \\
{\bf torus}  & $\{9,12\}$     & 0.094     & 0.383  \\
{\bf dmP}    & $\{8,10,12\}$  & 0.069     & 0.270  \\
\hline
\end{tabular}
\end{center}
\end{table}

We now present experimental results for larger graphs where we use the optimistic coupling.
All experiments were conducted on a computer with an
Intel(R) Xeon(R) CPU E5-2630 v3 @ 2.40GHz with 4 cores and 32GB of RAM.
We present running times for different graph topologies and sizes
in Figures~\ref{fig:densetime}, \ref{fig:eyestime}, \ref{fig:cycletime},
\ref{fig:sparsetime}, \ref{fig:2dstime}, \ref{fig:2dtime}  and \ref{fig:dmtime}.
Note beforehand that the coupling estimate needs only to be computed once for each graph.
Once the estimate is known, we can generate as many spanning trees as we want.
Although the edge swapping method is not always the faster compared
with the random walk and the Wilson's algorithm, it is competitive
in practice for {\bf dmP} and {\bf torus} graphs, and it is faster
for {\bf biK}, {\bf cycle} and {\bf sparse} (ladder) graphs. As expected,
it is less competitive for {\bf dense} graphs.
Hence, experimental results seem to point out that the edge
swapping method is more competitive in practice for those
instances that are harder for random walk based methods, namely {\bf biK}
and {\bf cycle} graphs.
The results for {\bf biK} and {\bf dmP} are of particular interest as most
real networks seem to include these kind of topologies, i.e., they include
communities and they are scale-free~\cite{chung2006complex}.

\begin{figure}[tbp]
  \begin{center}
    \scalebox{0.75}{        \input{dense_time.tex}}
  \end{center}
\caption{Running times for {\bf dense} (fully connected) graphs averaged over five runs, including the running time for computing the optimistic coupling estimate, the running time for generating a spanning tree based on that estimate and on the edge swapping algorithm, the running time for generating a spanning tree through a random walk, and also the running time for Wilson's algorithm.}
  \label{fig:densetime}
\end{figure}
\begin{figure}[tbp]
  \begin{center}
    \scalebox{0.75}{        \input{eyes_time.tex}}
  \end{center}
\caption{Running times for {\bf biK} graphs averaged over five runs, including the running time for computing the optimistic coupling estimate, the running time for generating a spanning tree based on that estimate and on the edge swapping algorithm, the running time for generating a spanning tree through a random walk, and also the running time for Wilson's algorithm.}
  \label{fig:eyestime}
\end{figure}
\begin{figure}[tbp]
  \begin{center}
    \scalebox{0.75}{        \input{cycle_time.tex}}
  \end{center}
\caption{Running times for {\bf cycle} graphs averaged over five runs, including the running time for computing the optimistic coupling estimate, the running time for generating a spanning tree based on that estimate and on the edge swapping algorithm, the running time for generating a spanning tree through a random walk, and also the running time for Wilson's algorithm.}
  \label{fig:cycletime}
\end{figure}
\begin{figure}[tbp]
  \begin{center}
    \scalebox{0.75}{        \input{sparse_time.tex}}
  \end{center}
\caption{Running times for {\bf sparse} (ladder) graphs averaged over five runs, including the running time for computing the optimistic coupling estimate, the running time for generating a spanning tree based on that estimate and on the edge swapping algorithm, the running time for generating a spanning tree through a random walk, and also the running time for Wilson's algorithm.}
  \label{fig:sparsetime}
\end{figure}
\begin{figure}[tbp]
  \begin{center}
    \scalebox{0.75}{        \input{2ds_time.tex}}
  \end{center}
\caption{Running times for square {\bf torus} graphs averaged over five runs, including the running time for computing the optimistic coupling estimate, the running time for generating a spanning tree based on that estimate and on the edge swapping algorithm, the running time for generating a spanning tree through a random walk, and also the running time for Wilson's algorithm.}
  \label{fig:2dstime}
\end{figure}
\begin{figure}[tbp]
  \begin{center}
    \scalebox{0.75}{        \input{2d_time.tex}}
  \end{center}
\caption{Running times for rectangular {\bf torus} graphs averaged over five runs, including the running time for computing the optimistic coupling estimate, the running time for generating a spanning tree based on that estimate and on the edge swapping algorithm, the running time for generating a spanning tree through a random walk, and also the running time for Wilson's algorithm.}
  \label{fig:2dtime}
\end{figure}
\begin{figure}[tbp]
  \begin{center}
    \scalebox{0.75}{        \input{dm_time.tex}}
  \end{center}
\caption{Running times for {\bf dmP} graphs averaged over five runs, including the running time for computing the optimistic coupling estimate, the running time for generating a spanning tree based on that estimate and on the edge swapping algorithm, the running time for generating a spanning tree through a random walk, and also the running time for Wilson's algorithm.}
  \label{fig:dmtime}
\end{figure}

\section{Related Work}
\label{sec:related-work}

For a detailed exposure on probability on trees and networks
see~\citet*[Chapter 4]{lyons2016probability}. As far as we know, the initial
work on generating uniform spanning trees was by
\citet*{Aldous:1990:RWC:87414.87417} and \citet*{Broader89}, which obtained
spanning trees by performing a random walk on the underlying graph. The
author also further studied the properties of such random
trees~\citep{RSA:RSA3240010402}, namely giving general closed formulas for
the counting argument we presented in Section~\ref{sec:problem}. In the
random walk process a vertex $v$ of $G$ is chosen and at each step this
vertex is swapped by an adjacent vertex, where all neighboring vertices are
selected with equal probability. Each time a vertex is visited by the first
time the corresponding edge is added to the growing spanning tree. The
process ends when all vertexes of $G$ get visited at least once. This
amount of steps is known as the cover time of $G$.

To obtain an algorithm that is faster than the cover time, 
\citet*{Wilson:1996:GRS:237814.237880} proposed a different approach. A
vertex $r$ of $G$ is initially chosen uniformly and the goal is to hit this
specific vertex $r$ from a second vertex, also chosen uniformly from
$G$. This process is again a random walk, but with a loop erasure
feature. Whenever the path from the second vertex intersects itself all the
edges in the corresponding loop must be removed from the path. When the
path eventually reaches $r$ it becomes part of the spanning tree. The
process then continues by choosing a third vertex and also computing a loop
erasure path from it, but this time it is not necessary to hit $r$
precisely, it is enough to hit any vertex on the branch that is already
linked to $r$. The process continues by choosing random vertexes an
computing loop erasure paths that hit the spanning tree that is already
computed.

We implemented the above algorithms, as they are accessible, although
several theoretical results where obtained in recent years we are not aware
of an implementation of such algorithms. We will now survey these results.

Another approach to this problem relies on the Kirchoff's
Theorem~\citet*{ANDP:ANDP18471481202} that counts the number of spanning
trees by computing the determinant of a certain matrix, related to the
graph $G$. This relation is researched
by~\citet*{Guenoche,KULKARNI1990185}, which yielded an $O(E V^3)$ time
algorithm. This result was improved to $O(V^{2.373})$, by
\citet*{DBLP:journals/jal/ColbournDN89,DBLP:journals/jal/ColbournMN96},
where the exponent corresponds to the fastest algorithm to compute matrix
multiplication. Improvements on the random walk approach where obtained
by~\citet*{5438651,mkadry2011graphs}, culminating in an
$\tilde{O}(E^{o(1)+4/3})$ time algorithm
by~\citet*{doi:10.1137/1.9781611973730.134}, which relies on insight
provided by the effective resistance metric.

Interestingly, the initial work by \citet*{Broader89} contains reference to
the edge swapping chain we presented in this paper (Section 5, named the
swap chain). The author mentions that the mixing time of this chain is
$E^{O(1)}$, albeit the details are omitted. As far as we can tell, this
natural approach to the problem did not receive much attention precisely
due to the lack of an efficient implementation. Even though link cut trees
were known at the
time~\citep*{Sleator:1981:DSD:800076.802464,Sleator:1985:SBS:3828.3835}
their application to this problem was not established prior to this
work. Their initial application was to network
flows~\citep{Goldberg:1989:FMC:76359.76368}. We also found another reference
to the edge swap in the work of~\citet*{sinclair1992improved}. In the
proposal of the canonical path technique the author mentions this
particular chain as a motivating application for the canonical path
technique, still the details are omitted and we were not able to obtain
such an analysis.

We considered the LCT version where the auxiliary trees are implemented
with splay trees~\citet{Sleator:1985:SBS:3828.3835}, i.e., the auxiliary
data structures we mentioned in Section~\ref{sec:idea} are splay
trees. This means that in step 5 of Algorithm~\ref{randomizeStep} all the
vertexes involved in the path $C \setminus \{(u,v)\}$ get stored in a splay
tree. This path oriented approach of link cut trees makes them suitable for
our goals, as opposed to other dynamic connectivity data structures such as
Euler tour trees~\citep{henzinger1995randomized}.

Splay trees are self-adjusting binary search trees, therefore the vertexes
are ordered in such a way that the inorder traversal of the tree coincides
with the sequence of the vertexes that are obtained by traversing
$C \setminus \{(u,v)\}$ from $u$ to $v$. This also justifies why the size
of this set can also be obtained in $O(\log V)$ amortized time. Each node
simply stores the size of its sub-tree and these values are efficiently
updated during the splay process, which consists of a sequence of
rotations. Moreover, these values can also be used to \texttt{Select} an
edge from the path. By starting at the root and comparing the tree sizes to
$i$ we can determine if the first vertex of the desired edge is on the left
sub-tree, on the root or on the right sub-tree. Likewise we can do the same
for the second vertex of the edge in question. These operations splay the
vertexes that they obtain and therefore the total time depends on the
\texttt{Splay} operation. The precise total time of the \texttt{Splay}
operation is $O((V +1) \log n)$, however the $V \log V$ term does not
accumulate over successive operations, thus yielding the bound of
$O((V + \tau ) \log V)$ in Theorem~\ref{teo1}. In general the $V \log V$
term should not be a bottleneck because for most graphs we should have
$\tau > V$. This is not always the case, if $G$ consists of a single cycle
then $\tau = 1$, but $V$ may be large. Figure~\ref{fig:cycleex} shows an
example of such a graph.

We finish this Section by reviewing the formal definitions of variational
distance and mixing time $\tau$~\citep*{Mitzenmacher:2005:PCR:1076315}.
\begin{definition}
  The variation distance between two distributions $D_1$ and $D_2$ on  a
  countable space $S$ is given by
  \begin{equation}
    ||D_1 - D_2|| = \sum_{x \in S} \frac{|D_1(x)-D_2(x)|}{2}
  \end{equation}
\end{definition}
\begin{definition}
  Let $\pi$ be the stationary distribution of a Markov chain with state
  space $S$. Let $p_x^t$ represent the distribution of the state of the
  chain starting at state $x$ after $t$ steps. We define
  \begin{equation}
    \Delta_x(t) = |p_x^t - \pi|
  \end{equation}
  \begin{equation}
    \Delta(t) = \max_{x \in S} \Delta_s(t)
  \end{equation}
That is $ \Delta_x(t)$ is the variation distance between the stationary
distribution and $p_x^t$ and $\Delta(t)$ is the maximum of these values
over all states $x$. We also define
  \begin{equation}
    \tau_x(\varepsilon) = min \{t:\Delta_x(t) \leq \varepsilon\}
  \end{equation}
  \begin{equation}
    \tau(\varepsilon) = \max_{x \in S} \tau_x(\varepsilon)
  \end{equation}
\end{definition}
When we refer only to the mixing time we mean $\tau(1/4)$. Finally the
coupling Lemma justifies the coupling approach:
\begin{lemma}
\label{lem:coupling}
  Let $Z_t = (X_t, Y_t)$ be a coupling for a Markov chain $M$ on a state
  space $S$. Suppose that there exists a $T$ such that, for every $x,y \in
  S$,
  \begin{equation}
    \Pr(X_T \neq Y_T | X_0 = x, Y_0 = y) \leq \varepsilon
  \end{equation}
  Then $\tau(\varepsilon) \leq T$. That is, for any initial state, the
  variation distance between the distribution of the state of the chain
  after $T$ steps and the stationary distribution is at most $\varepsilon$.
\end{lemma}
If there is a distance $d$ defined in $S$ then the property $X_T \neq Y_T$
can be obtained using the condition $d(X_T, Y_T) \geq 1$. For this
condition we can use the Markovian inequality
$\Pr(d(X_T, Y_T) \geq 1) \leq E[d(X_T, Y_T)]$. The path coupling
technique~\citep*{646111} constructs a coupling by chaining several chains,
such that the distance between then is 1. Therefore we obtain
$d(X_T, Y_T) = d(X_T^0,X_T^1) + d(X_T^1,X_T^2) + \ldots +
d(X_T^{D-1},X_T^D) = 1 + 1 + \ldots + 1$, where $X_T = X_T^0$ and
$Y_T = X_T^D$.
\section{Conclusions and Future Work}
\label{sec:conclusions}

In this paper we studied a new algorithm to obtain the spanning trees of a
graph in an uniform way. The underlying Markov chain was initially sketched
by~\citet*{Broader89} in the early study of this problem. We further extended
this work by proving the necessary Markov chain properties and using the
link cut tree data structure. This allows for a much faster implementation
than repeating the DFS procedure. This may actually be the reason why
this approach has gone largely unnoticed during this time.

The main shortcoming of our approach is the lack of a general theoretical
bound of the mixing time. Such a bound might be possible using new
approaches as the insight into the resistance
metric~\citep*{doi:10.1137/1.9781611973730.134}. Although a general
analysis would be valuable we addressed this problem by simulating a
coupling, both sound or optimistic. The lack of analysis is not a
shortcoming of the algorithm itself, which is both practical and
efficient. We implemented it and compared it against existing
alternatives. The experimental results show that it is very competitive. A
theoretical bound would still be valuable, specially if it is the case that
this algorithm is more efficient than the alternatives.

On the one hand computing the mixing time of the underlying chain is
complex, time consuming and hard to analyse in theory. On the other hand
the user of this process can fix a certain number of steps to execute. This
is a very useful parameter, as it can be used to swap randomness for
time. Depending on the type of application the user may sacrifice the
randomness of the underlying trees to obtain faster results or on the
contrary spend some extra time to guarantee randomness. Existing algorithms
do not provide such a possibility.

As a final note, we point out that our approach can be generalized by
assigning weights to the edges of the graph. The edge to be inserted can
then be selected with a probability that corresponds to its weight, divided
by the global sum of weights. Moreover the edge to remove from the cycle
should be removed according to its weight. The probability should be its
weight divided by the sum of the cycle weights. The ergodic analysis of
Section~\ref{sec:analysis} generalizes easily to this case, so this chain
also generates spanning trees uniformly, albeit the analysis of the
coupling of Section~\ref{sec:coping} might need some adjustments. A proper
weight selection might obtain a faster mixing timer, possibly something
similar to the resistance of the edge.

\section*{Acknowledgements}
This work was funded in part by
  European Union's Horizon 2020 research and innovation programme
  under the Marie Sk{\l}odowska-Curie grant agreement No 690941 and
  national funds through Funda\c{c}\~{a}o para a Ci\^{e}ncia e a
  Tecnologia (FCT) with reference UID/CEC/50021/2013.
%



\section*{Bibliography}
\bibliographystyle{elsarticle-harv}
\bibliography{elsarticle-bibliography}






\end{document}